\documentclass[journal]{IEEEtran}
\usepackage[numbers]{natbib}
\usepackage{amsmath,amsfonts}
\usepackage{array,algorithmic}
\usepackage{textcomp}
\usepackage{multirow}
\usepackage{stfloats}
\usepackage{url}
\usepackage{verbatim}
\usepackage{graphicx}
\usepackage{textcomp}
\usepackage{subfigure}
\usepackage{mathrsfs}
\usepackage{pifont}
\usepackage{diagbox}
\usepackage{booktabs}
\usepackage{amsmath}
\usepackage{amsthm}
\usepackage[linesnumbered,ruled,vlined]{algorithm2e}
\newtheorem{definition}{Definition}
\newtheorem{lemma}{Lemma}

\newtheorem{proposition}{Proposition}
\newtheorem{corollary}{Corollary}
\SetKwInput{KwInput}{Input}
\SetKwInput{KwOutput}{Output} 
\ifCLASSINFOpdf
\else
\fi
\begin{document}
\title{Complex Network Modelling with Power-law Activating Patterns and Its Evolutionary Dynamics}
\author{Ziyan~Zeng,~\IEEEmembership{Student Member,~IEEE,}
        Minyu~Feng,~\IEEEmembership{Senior Member,~IEEE,}
        Pengfei~Liu,
        and~J\"{u}rgen~Kurths% <-this % stops a space
\thanks{Ziyan Zeng and Minyu Feng are with the College of Artificial Intelligence, Southwest University, Chongqing 400715, China (e-mail: myfeng@swu.edu.cn). }% <-this % stops a space
\thanks{Pengfei Liu is with the School of Computer Science and Engineering,
University of Electronic Science and Technology of China, and Yangtze Delta Region Institute (Quzhou), University of Electronic Science and Technology of China. }
\thanks{J\"{u}rgen Kurths is with the Potsdam Institute for Climate Impact Research, 14437 Potsdam, Germany, and the Institute of Physics, Humboldt University of Berlin, Berlin, Germany. }% <-this % stops a space
\thanks{Manuscript received , 2024; revised , 2024. This work was supported by the National Natural Science Foundation of China (NSFC) (Grant No. 62206230), the Natural Science Foundation of Chongqing (Grant No. CSTB2023NSCQ-MSX0064), and the Graduate Research Innovation Project of Southwest University (Grant No. SWUS24190). (Corresponding author: Minyu Feng.)}}

\markboth{IEEE TRANSACTIONS ON SYSTEMS, MAN, AND CYBERNETICS: SYSTEMS, ~Vol. , No. , August~2023}%
{Zeng et al. : Complex Networks with Power-law Activating Patterns and Dynamics}
\maketitle
\begin{abstract}
\small
Complex network theory provides a unifying framework for the study of structured dynamic systems. The current literature emphasizes a widely reported phenomenon of intermittent interaction among network vertices. In this paper, we introduce a complex network model that considers the stochastic switching of individuals between activated and quiescent states at power-law rates and the corresponding evolutionary dynamics. By using the Markov chain and renewal theory, we discover a homogeneous stationary distribution of activated sizes in the network with power-law activating patterns and infer some statistical characteristics. To better understand the effect of power-law activating patterns, we study the two-person-two-strategy evolutionary game dynamics, demonstrate the absorbability of strategies, and obtain the critical cooperation conditions for prisoner's dilemmas in homogeneous networks without mutation. The evolutionary dynamics in real networks are also discussed. Our results provide a new perspective to analyze and understand social physics in time-evolving network systems. 
\end{abstract}
\IEEEpeerreviewmaketitle
\section{Introduction}\label{sec: I}
\small
Since the development of the random graph theory, an increasing number of researchers have been paying attention to network science. At the end of the last century, \cite{watts1998collective} and \cite{barabasi1999emergence} proposed the small-world and scale-free network models, speeding up the development of network science \cite{albert2002statistical}. Generally, the complex network theory provides a framework for the study of individual relationships and dynamic processes in structured populations \cite{boccaletti2006complex}. In the past few years, more specific network concepts, e.g., temporal \cite{holme2015modern} and higher-order networks \cite{benson2016higher}, have been employed to meet the needs of studying social physics \cite{jusup2022social}, including evolutionary dynamics \cite{li2020evolution,pi2024memory} and epidemic propagation \cite{li2023coevolution}. Besides, some articles noticed that some real-world networks do not simply show growth and preferential attachment characteristics of vertices, but the death of existing vertices is also observed as a common phenomenon \cite{zhang2016random, feng2024information}. Feng {\it et al.} \cite{feng2022heritable} noticed that the random birth \& death of individuals breaks the scale-free property induced by preferential attachment and further proposed the heritable mechanism to maintain the power-law characteristic of the evolving birth \& death network on continuous time stamps. Besides, the phase transitions in corresponding dynamic processes have also been shown to be strongly related to the birth \& death property \cite{zeng2022spatial,li2021protection}. The nonlinear dynamic research in complex networks reveals the phase transition \cite{biroli2007new} and collective behavior \cite{da2023sociophysics} of dynamic processes in structured populations, such as synchronization \cite{lu2019quad}, the evolutionary game \cite{guo2023third, yao2023inhibition}, spreading \cite{ji2023signal}, and percolation \cite{li2021percolation}. Additionally, the mathematical theorems of complex networks to describe and control diverse real network systems are essential as well \cite{reed2023role,kyriakis2020effects,yang2021hidden}. 

In the study of dynamic processes, the inter-event time between two adjacent events during the interaction is significant in the theoretical analysis of population behaviors, especially in the temporal networks and dynamics. In this circumstance, an individual does not always interact with an arbitrary neighbor, as this interaction is sometimes activated but dormant at other time. A common assumption is that individual actions are uniformly distributed over a period of time and thus lead to the Poisson statistics \cite{miritello2011dynamical,Zeng2023temporal}. However, the empirical evidence has shown that the inter-event time of individuals in a system can follow a power-law distribution and result in heavy-tailed statistics \cite{barabasi2005origin}. For example, Roberts {\it et al.} \cite{roberts2015heavy} found that neuronal oscillations exhibit non-Gaussian heavy-tailed probability distributions. Malmgren {\it et al.} \cite{malmgren2008poissonian} demonstrated that the approximate power-law scaling of the inter-event time distribution is a consequence of circadian and weekly cycles of human activity, which explains the e-mail communication data precisely. Han {\it et al.} \cite{han2023impact} found that the heterogeneity of activation time in time-varying networks significantly affects the spreading threshold. A more comprehensive empirical analysis of the heavy tails in human dynamics is given by V{\'a}zquez {\it et al.} \cite{vazquez2006modeling}, and found that the inter-event time in web browsing, email activity patterns, library loans, trade transactions, and the correspondence patterns of Einstein, Darwin, and Freud can all be well approximated by a power-law distribution with different parameters. Recent studies in brain dynamics \cite{wang2019non,lombardi2020critical} also show that given certain correlations or anti-correlations between consecutive bursts in $\theta$ and $\delta$, cortical rhythms exhibit complex temporal organization, leading to power-law distributed duration for active states. Based on the power-law inter-event assumption, researchers have studied the dynamic processes including epidemic propagation \cite{min2011spreading} and random walks \cite{masuda2017random}. Considering the inter-event time on both the nodes and edges, \cite{dos2020generative} proposed a generative model that produces distributions of inter-event times for both nodes and edges that resemble heavy-tailed distributions across some scales. 

This study mainly focuses on the network modelling of the aforementioned power-law activation and its effect on evolutionary dynamics, which reveals the formation and fixation of individual strategies \cite{perc2017statistical}. Since the introduction of spatial chaos in evolutionary game theory \cite{nowak1992evolutionary}, it has become widely accepted that spatial structures enhance the emergence of cooperative behaviors in society \cite{nowak1994more}. Later, it was found that the spatial structure often inhibits the emergence of cooperation \cite{hauert2004spatial}, which challenges the conventional understanding of the relationship between complex networks and evolutionary games. \cite{santos2005scale} studied the weak prisoner's dilemma and snowdrift game based on scale-free network theory and found that scale-free networks significantly improve cooperation density through the influence of hub vertices \cite{santos2006evolutionary}. To uncover the condition for cooperation, \cite{ohtsuki2006simple} used the pair approximation method and discovered that population preference for cooperation is strongly correlated with the average number of neighbors. However, this condition is not particularly effective in heterogeneous populations, prompting \cite{konno2011condition} to identify a closer condition for cooperation in scale-free networks. Considering the concept of stochastic game theory, an evolutionary game model with game transitions was proposed by introducing the switching of high and low games, which helps to overcome social dilemmas \cite{su2019evolutionary, 10269140}. In addition to the above studies, recent articles primarily focus on the evolutionary dynamics of temporal networks \cite{li2020evolution}, providing an innovative avenue for emergent behaviors. 

As mentioned previously, the inter-event time of individuals' activities in a complex network is essential for both network structure and dynamic processes. However, to the best of our knowledge, the effect of stochastic inter-event time between activities on closed network structure and topological properties has been largely overlooked, as well as its evolutionary dynamics. Additionally, most real systems show the power-law intermittent properties among individuals as stated previously. Therefore, based on the assumption of power-law inter-event time, this study aims to model temporal networks with power-law intermittent activities and to provide a new perspective for understanding power-law activating patterns and evolutionary game theory in closed, networked populations. 

Therefore, we assume that there are two switching states for vertices in a network structure, namely the activated and quiescent states. Additionally, for each individual, these two states switch at random intervals that follow a power-law distribution, based on the aforementioned findings in real-world data sets. As shown in Ref. \cite{feng2018evolving}, we consider the number of activated vertices a topological property of the network. The theoretical analysis primarily relies on Markov chain and renewal theory \cite{cinlar1969markov}.

The evolution of cooperation in networked systems reveals the mechanisms underlying the propagation of altruistic behaviors, which directly enhances the overall fitness of the population. To measure the impact of power-law activation patterns on the spatial evolution of cooperation, we study the prisoner's dilemma game \cite{axelrod1980effective} by incorporating the power-law intermittent interactions of individuals into both modelling and simulation. We assume that only activated vertices participate in the spatial evolutionary game, including interacting with each other, obtaining fitness, and updating strategies. An activated vertex gains payoff from all its activated neighbors and translates the total payoff into fitness. Various strategy updating rules can be considered to drive the evolutionary process of individual strategies. In this article, we employ the death-birth process \cite{kaveh2015duality}, which assumes that individuals with higher payoffs are more likely to spread their strategies to their neighbors.

The main contributions of this paper are as follows:
\begin{enumerate}
\item A complex network model incorporating individuals' power-law activating patterns is introduced to describe the state switching of vertex activation. Theoretical analysis using renewal theory is conducted to prove the stability of activated individuals.
\item The evolutionary dynamics of the proposed power-law activating patterns in network structures are analyzed. The absorption of cooperative behaviors in a closed network structure is demonstrated, and the critical condition for cooperation in the prisoner's dilemma game is derived.
\item The proposed theorems are validated through comparative experiments. The fixation probability of evolutionary dynamics is examined through computer simulations. Furthermore, applications in real-world network data sets are discussed.
\end{enumerate}

The organization of this article is as follows: In Sec. \ref{sec: II}, we introduce the complex network model with power-law activating patterns and describe the evolutionary dynamics. In Sec. \ref{Sec: III}, we present the simulation methods and conduct experiments to validate our propositions. In Sec. \ref{Sec: IV}, we conclude our work and provide suggestions for future research.
\section{Network Model and Evolutionary Dynamics Analysis}\label{sec: II}
\small
As stated previously, individuals in a networked system often exhibit intermittent activity patterns, where the inter-event time is stochastic and follows a specific probability distribution, typically a power-law distribution. In this section, we introduce our complex network model with power-law activating patterns and analyze its evolutionary game dynamics to examine the impact of intermittent activity patterns on cooperation. Important notations are presented in Tab. \ref{tab: notations}.
\begin{table}[h]
\centering
\caption{Notations}\label{tab: notations}
\begin{tabular}{cc}
   \toprule
   Symbol & Definition\\
   \midrule
   $\lambda$ & The power-law parameter for quiescent states\\
   $\mu$ & The power-law parameter for activated states\\
   $\mathcal{G}=(\mathcal{V},\mathcal{E})$ & The underlying network\\
   $\mathcal{V}_1(t)$ & The activated vertex set\\
   $X_i(t)$ & The state of the vertex $i$, quiescent or activated\\
   $N$&The number of vertices in the underlying network \\
   $N_1(t)$ & The size of activated subgraph\\
   $M$ & The payoff matrix of two-player and two-strategy game \\
   $\pi_i(\mathbf{s})$ & The payoff of the vertex $i$\\
   $f_i(\mathbf{s})$ & The fitness of the vertex $i$ \\
   $w$ & The selection intensity \\
    $\rho_C$ ($\rho_D$)& The fixation probability of the cooperation (defection)\\
    $b/c$& The benefit-to-cost ratio of the prisoner's dilemma game\\
     $l_{ij}$&The one step random walk probability from the vertex $i$ to $j$ \\
    $D$& The instantaneous change of the cooperation proportion\\
    $\tau_{ij}$&The expected coalescence time from vertices $i$ and $j$\\
    $\tau^{(n)}$&The expectation of $\tau_{ij}$ of all $n$-step random walks\\
   \bottomrule
\end{tabular}
\end{table}
\subsection{Complex network with power-law activating patterns}\label{sec: II(A)}
\small
\subsubsection{Network Modelling Description}
\small
\begin{algorithm}
	\caption{Realization of Power-law Activating Patterns}
	\label{alg:1}
	\begin{algorithmic}[1]
		\STATE Initialize $\mathcal{G}$, $\mathcal{V}_1(0)$, $t=0$, and the max time $T$. 
            \STATE Initialize the next state transition time $\Gamma$ for each vertex. 
            \REPEAT
            \STATE Find the smallest time $t_{temp}$ in $\Gamma$ and the corresponding vertex $i$. 
            \IF{$i$ is activated}
            \STATE Turn $i$ into quiescent. 
            \STATE Sample $i$'s next state transition time to the activated state based on a power-law distribution with the parameter $\lambda$. 
            \ELSE
            \STATE Turn $i$ into activated. 
            \STATE Sample $i$'s next state transition time to the quiescent state based on a power-law distribution with the parameter $\mu$. 
            \ENDIF
            \UNTIL{$t\geq T$}
	\end{algorithmic}  
\end{algorithm}
In order to describe the network structure with both activated and quiescent vertices, we introduce the stochastic switching mechanism of a vertex between activated and quiescent states at power-law rates using a continuous-time Markov chain. We consider an unweighted, undirected network $\mathcal{G}=(\mathcal{V}, \mathcal{E})$ with vertex set $\mathcal{V}$, edge set $\mathcal{E}$, and $N$ vertices. We assume that the network $\mathcal{G}$ represents a closed population with no new arrivals or departures. As described previously, two types of vertices in $\mathcal{V}$ undergo a phase transition over time, including both activated and quiescent vertices. We denote the set of activated vertices (\textit{resp.}, the set of quiescent vertices) at time $t$ as $\mathcal{V}_1(t)$ (\textit{resp.} $\mathcal{V}_2(t)$), respectively.

Initially, there are $\mathcal{V}_1(0)$ activated vertices and $\mathcal{V}_2(0)$ quiescent vertices in the network $\mathcal{G}$. As mentioned, the intermittent behaviors of individuals typically follow a power-law inter-event time distribution \cite{barabasi2005origin, roberts2015heavy, malmgren2008poissonian, vazquez2006modeling}, e.g., email communications, neural activities, and human interactions. Therefore, we use two independent power-law distributions here to describe the alternating phase transitions of individual states. If an activated vertex becomes quiescent at a certain moment, it remains quiescent for a time $t$, which follows a continuous power-law distribution with the probability density function
\begin{small}
\begin{equation}\label{eq: power-law distribution}
f(t; \lambda)=\frac{\lambda-1}{t_0}(\frac{t}{t_0})^{-\lambda}, \lambda>2, 
\end{equation}
\end{small}
where $t_0 > 0$ is the lower bound of the power-law distribution. In contrast, if a quiescent vertex becomes activated at a certain moment, it remains activated for a time $t$ that follows the continuous power-law distribution $f(t; \mu)$ before becoming quiescent again, where $\mu > 2$. We refer to $\lambda$ and $\mu$ as the power-law rates for activation and quiescence, respectively. In this case, the state transition process of each individual is considered a cyclic but not strictly periodic process. In this paper, we simulate networks with power-law activating patterns using Algorithm \ref{alg:1}.

\begin{figure}                                              %fig1
\centering
\subfigure[State Transition of a Single Vertex]{\includegraphics[scale=0.17]{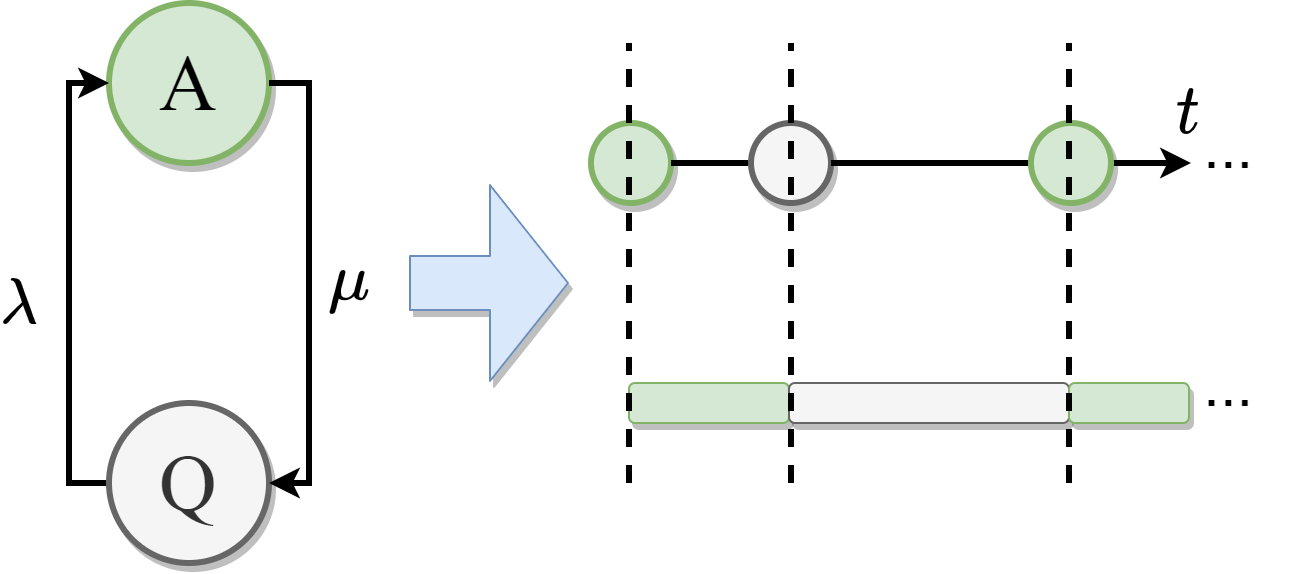}}
\vspace{-3.5mm}
\subfigure[Evolving Network Structure]{\includegraphics[scale=0.09]{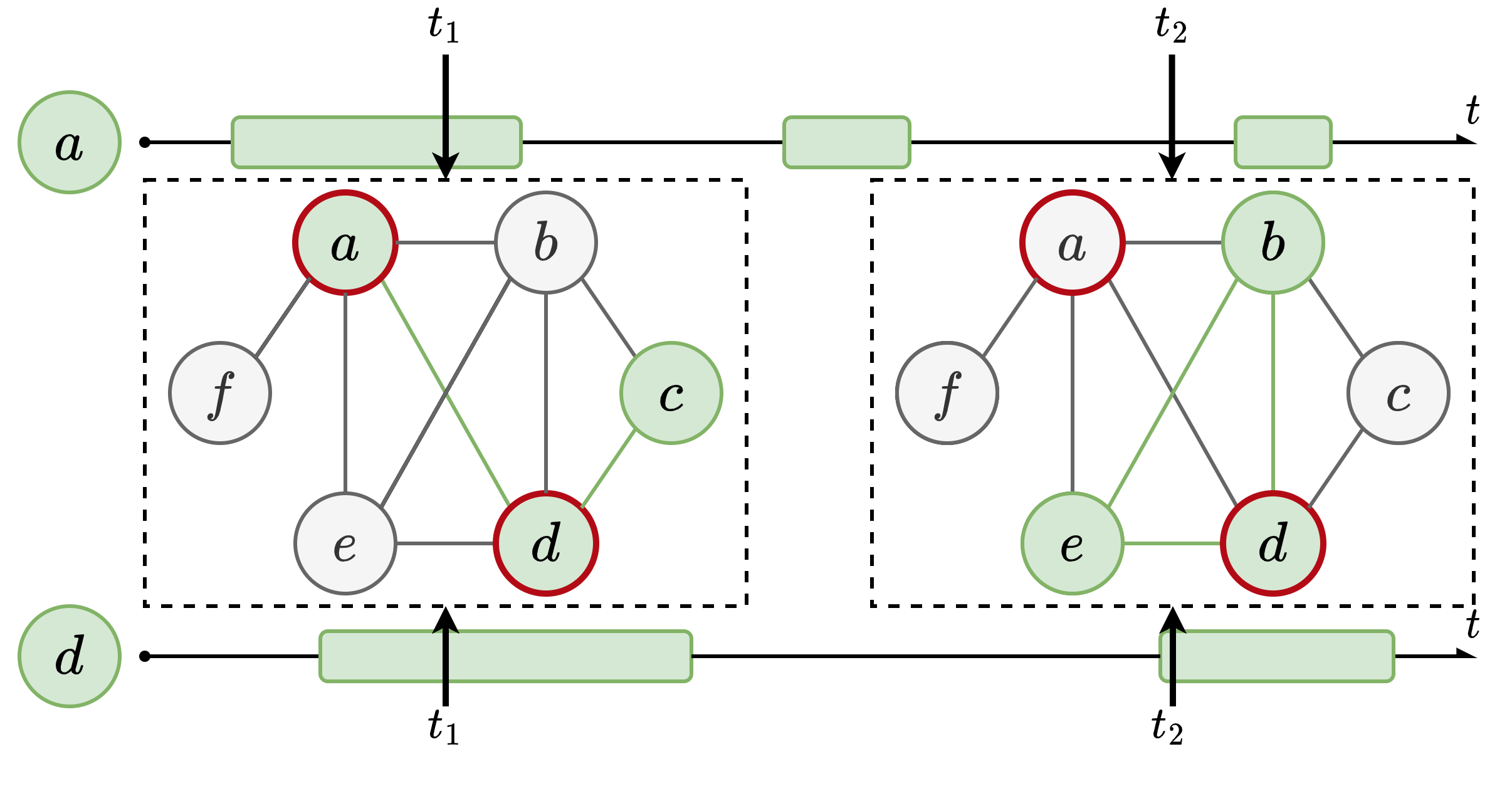}}

\caption{\textbf{An example of the proposed network model. } Green and grey vertices are in activated and quiescent states respectively. (a) The state transition of a single vertex with power-law activating patterns. An activated vertex becomes quiescent after a power-law period with the parameter $\mu$, while a quiescent vertex turns activated after another power-law period with the parameter $\lambda$. (b) The evolving of network structure with the power-law activating patterns of vertices. We consider a network with six vertices and nine edges. We show the state transition time stamps of nodes $a$ and $d$ and the network snapshots at $t_1$ and $t_2$. Green periods indicate the activated duration. If both ends of an edge are activated, this edge is then activated. (Color online)}\label{fig: 1}
\end{figure}

For a better presentation, an example of the proposed network model with power-law activating patterns is shown in Fig. \ref{fig: 1}. 

\subsubsection{Theoretical Analysis of the Activated Subgraph}
For convenience, we introduce the following model definitions. 
\begin{definition}
The activated subgraph at time $t$ is $\mathcal{G}_1(t) = (\mathcal{V}_1(t), \mathcal{E}_1(t))$, where two ends of each element in $\mathcal{E}_1(t)$ are both activated vertices. The quiescent subgraph at time $t$ is $\mathcal{G}_2(t) = (\mathcal{V}_2(t), \mathcal{E}_2(t))$, where one end of each element in $\mathcal{E}_2(t)$ is a quiescent vertex. 
\end{definition}
Apparently, we have $\mathcal{G}=\mathcal{G}_1(t)\cup \mathcal{G}_2(t)$. Then, studying the properties of $\mathcal{G}_1(t)$ can easily induce the properties of $\mathcal{G}_2(t)$. In our work, we are particularly concerned about the activated subgraph size. Therefore, we denote the following stochastic process to describe the activated scale of the network. 
\begin{definition}
$\{X_i(t), t\geq 0, i\in \mathcal{V}\}$ is a continuous time stochastic process with the state space $\Omega_i=\{0, 1\}$ that describes the state of the individual $i$, where $0$ and $1$ denote the quiescent and activated states respectively. 
\end{definition}
\begin{definition}
$\{N_1(t), t\geq 0\}$ is a continuous time stochastic process with the state space $\Omega=\{0, 1, 2, \cdots, N\}$ that describes the number of activated vertices, i.e., $N_1(t)=\vert \mathcal{V}_1(t)\vert$. 
\end{definition}
Apparently, $N_1(t)$ is the number of individuals that satisfies $X_i(t)=1$. Besides, the stochastic processes $X_i(t)$ and $N_1(t)$ are continuous-time Markov chains. Before we do further analysis, we introduce the following lemmas. 
\begin{lemma}\label{lemma: 1}
A recurrent Markov chain is regenerative. 
\end{lemma}
A regenerative stochastic process is allowed to restart at some random time. Once the process restarts, we say that this regenerative process undergoes a cycle since the last restart. The following lemma shows the long-range property of a regenerative process. 
\begin{lemma}\label{lemma: 2}
For a regenerative process $X(t)$, the limiting probability to find the process in state $i$ is \cite{ross2014introduction}
\begin{small}
\begin{equation}
\lim_{t\rightarrow\infty} P\{X(t)=i\}=\frac{E[T_i]}{E[T_c]},
\end{equation}
\end{small}
where $E[T_i]$ is the expected time that the process stays in the state $i$ during the cycle, and $E[T_c]$ denotes the expected cyclic time of $X(t)$. 
\end{lemma}

Based on Lemmas \ref{lemma: 1} and \ref{lemma: 2}, we next carry out the following proposition on the activated subgraph. 
\begin{proposition}\label{proposition: 1}
If a network $\mathcal{G}$ with $N$ vertices has power-law rates $\lambda$ and $\mu$ to be activated and quiescent respectively, the stationary distribution of the stochastic process $N_1(t)$ is irrelevant to the initial state and follows 
\begin{small}
\begin{equation}\label{eq: proposition 1}
P_i=\frac{C_N^i [(\mu-1)(\lambda-2)]^i [(\lambda-1)(\mu-2)]^{N-i}}{[(\mu-1)(\lambda-2)+(\lambda-1)(\mu-2)]^N}, 
\end{equation}
\end{small}
where $P_i$ denotes the probability of finding $i$ activated individuals in the stationary state and $C_N^i$ is the combinatorial operator. 
\end{proposition}
\begin{proof}
According to our previous model description, for an arbitrary vertex $i$, the stochastic process $X_i(t)$ is recurrent because both states are repeated infinitely as $t\rightarrow\infty$. Therefore, according to Lemma \ref{lemma: 1}, the process $X_i(t)$ is regenerative. 

A cycle contains two periods, including both activated and quiescent periods. Therefore, for an arbitrary vertex, the expectation of the quiescent time in a cycle is 
\begin{small}
\begin{equation}
E[T_0]=\int_{t_0}^{+\infty} t f(t;\lambda)\mathrm{d} t=\frac{\lambda-1}{\lambda-2}t_0.
\end{equation}
\end{small}
Similarly, the expectation of the activated time in a cycle is $E[T_1]=\frac{\mu-1}{\mu-2}t_0$. Then, according to Lemma \ref{lemma: 2}, in the long-range time, we can find a single vertex in the quiescent state with the probability
\begin{small}
\begin{equation}
\lim_{t\rightarrow\infty} P\{X(t)=0\}=\frac{E[T_0]}{E[T_c]}=\frac{\frac{\lambda-1}{\lambda-2}}{{\frac{\mu-1}{\mu-2}+\frac{\lambda-1}{\lambda-2}}}, 
\end{equation}
\end{small}
and in the activated state with the probability 
\begin{small}
\begin{equation}
\lim_{t\rightarrow\infty} P\{X(t)=1\}=\frac{E[T_1]}{E[T_c]}=\frac{\frac{\mu-1}{\mu-2}}{{\frac{\mu-1}{\mu-2}+\frac{\lambda-1}{\lambda-2}}}. 
\end{equation}
\end{small}
Therefore, the probability to find $i$ activated vertices in the long-range time follows the binomial form
\begin{small}
\begin{equation}
\begin{aligned}
P_i&=C_N^i (\frac{E[T_1]}{E[T_c]})^i (\frac{E[T_0]}{E[T_c]})^{N-i}\\
&=C_N^i (\frac{\frac{\mu-1}{\mu-2}}{{\frac{\mu-1}{\mu-2}+\frac{\lambda-1}{\lambda-2}}})^i (\frac{\frac{\lambda-1}{\lambda-2}}{{\frac{\mu-1}{\mu-2}+\frac{\lambda-1}{\lambda-2}}})^{N-i}\\
&=\frac{C_N^i [(\mu-1)(\lambda-2)]^i [(\lambda-1)(\mu-2)]^{N-i}}{[(\mu-1)(\lambda-2)+(\lambda-1)(\mu-2)]^N}. 
\end{aligned}
\end{equation}
\end{small}
Results follow. 
\end{proof}
Proposition \ref{proposition: 1} shows that the number of activated vertices is stable with power-law activating patterns in a closed population and gives the probability distribution of the activated subgraph size of a network $\mathcal{G}$ with given $\lambda$ and $\mu$. Accordingly, we have the following corollaries for further analysis and simulation. 
\begin{corollary}\label{corollary: 2}
The expected number of activated vertices is 
\begin{small}
\begin{equation}
E[N_1]=\frac{N(\mu-1)(\lambda-2)}{(\mu-1)(\lambda-2)+(\lambda-1)(\mu-2)}, 
\end{equation}
\end{small}
and the variance is 
\begin{small}
\begin{equation}
D[N_1]=\frac{N(\mu-1)(\lambda-2)(\lambda-1)(\mu-2)}{[(\mu-1)(\lambda-2)+(\lambda-1)(\mu-2)]^2}. 
\end{equation}
\end{small}
\end{corollary}
Apparently, if we consider a subgraph of $\mathcal{G}$, the deduction in Proposition \ref{proposition: 1} still holds. Therefore, we have the following corollary. 
\begin{corollary}\label{corollary: 3}
For an arbitrary non-empty subgraph $\mathcal{G}'$ of $\mathcal{G}$ with the size $N'$, the number of activated individuals is $i$ with the probability
\begin{small}
\begin{equation}\label{eq: corollary 3}
P(N^{'},i)=\frac{C_{N'}^i [(\mu-1)(\lambda-2)]^i [(\lambda-1)(\mu-2)]^{N'-i}}{[(\mu-1)(\lambda-2)+(\lambda-1)(\mu-2)]^{N'}}. 
\end{equation}
\end{small}
\end{corollary}
Corollary \ref{corollary: 2} presents the moments of the activated subgraph size. Corollary \ref{corollary: 3} indicates the self-similarity of the entire network and the local parts. For example, consider a subgraph of $\mathcal{G}$ with a vertex and its $N^{'}$ neighbors, the expected number of activated neighbors follows binomial distribution and can be calculated according to Corollary \ref{corollary: 3}. Here we note that we assume the switching between activation and quiescent does not correlate for each vertex, i.e., the activation transition of vertex $x$ does not affect the vertex $y$. This conclusion on the local topology of one single vertex helps analyze dynamic processes. This ends the network modelling section. 
\subsection{Spatial evolutionary game dynamics}\label{sec: II(B)}

As shown in many previous articles, network structures have fundamental influences on the evolution of cooperation. In addition, in time-varying networks, the change of network structure becomes even more crucial. Based on previous assumptions and conclusions on network structure with power-law activation patterns, we next analyze the evolutionary dynamics to show the impact of power-law intermittent behaviors among individuals on cooperation evolution.
\subsubsection{Evolutionary Game Model in Networks}
We consider the evolutionary game dynamics with two strategies on the complex network with individuals' power-law activating patterns. The strategies of individuals are cooperation (C) and defection (D). If two individuals play the game with each other, mutual cooperation brings each individual the reward $R$, while mutual defection brings each individual the punishment $P$. If only one individual cooperates, the defector receives the temptation $T$, while the cooperator receives the sucker's payoff $S$. Accordingly, we have the payoff matrix of a game between two players as follows. 

\begin{small}
\begin{equation}\label{eq: payoff}
    M=\left(
    \begin{array}{cc}
    R & S \\
    T & P \\
    \end{array}
\right). 
\end{equation}
\end{small}

We assume the most representative prisoner's dilemma game model (PDG), where $T>R>P>S$. In a network $\mathcal{G}$, each vertex is regarded as an individual in the evolutionary game and plays PDG with its neighbors. Note that in this paper, considering the power-law activating order of individuals, a player only interacts with its activated neighbors. That is, the strategy dynamics only occur in the activated subgraph $\mathcal{G}_1(t)$ at time $t$. Define a vector $\mathbf{s}=(s_i)_{i\in \mathcal{G}}\in\{0,1\}^{\mathcal{G}}$ that describes the strategy of each vertex in the network, where 1 and 0 present the cooperator and defector respectively. For an individual $i$, it derives the payoff from all its activated neighbors, denoted as
\begin{small}
\begin{equation}
    \pi_i(\mathbf{s})=Rs_is_i^{(1)}+Ss_i(1-s_i^{(1)})+T(1-s_i)s_i^{(1)}+P(1-s_i)(1-s_i^{(1)}), 
\end{equation}
\end{small}
where $s_i^{(n)}$ denotes the expected strategy of the vertex from $n$-step random walk away from the focal individual $i$. 
Then, we translate the payoff into fitness through $f_i(\mathbf{s})=1+w\pi_i(\mathbf{s})$, where $w\geq 0$ indicates the intensity of selection. We are especially concerned about weak selection, i.e., $0<w\ll 1$. 

Regarding strategy updating, we assume that only activated individuals update their strategies during the evolution process. We then introduce a Poisson process to illustrate the updating order due to the continuous-time setting. If an arbitrary individual becomes activated at some moment, it updates its strategy at a Poisson rate $\delta$; that is, the timestamps for strategy updates follow a Poisson process with parameter $\delta$ until the individual becomes quiescent. We use the Poisson process here since it has minimal impact on the results of cooperation evolution, allowing us to focus on the effects of power-law activation patterns. As mentioned previously, the expected activation time of an individual is $E[T_1] = \frac{\mu-1}{\mu-2}t_0$. The expected strategy update frequency of a vertex during activation is $\frac{(\mu-1)t_0}{(\mu-2)\delta}$. The strategy update rule we assume is similar to death-birth updating. If an arbitrary activated individual $x$ updates its strategy, each activated neighbor spreads its strategy to $x$ with a probability proportional to fitness. In this case, $x$'s activated neighbors with high fitness are more likely to spread their strategies to $x$. We note that, due to the continuous-time setting, the probability that at least two different vertices update their strategies at the same moment tends to zero and can be neglected. An example of the strategy update is shown in Fig. \ref{fig: 2}.

Additionally, mutation can exist in evolutionary dynamics. If we consider mutation when the strategy update event occurs, with probability $v$, the focal updating vertex becomes a cooperator or defector with equal probability regardless of the fitness, and with probability $1-v$, this vertex undergoes the previously mentioned strategy update process. 

\begin{figure}                                              %fig2
\centering
\subfigure[Strategy Update Based on Fitness]{\includegraphics[scale=0.11]{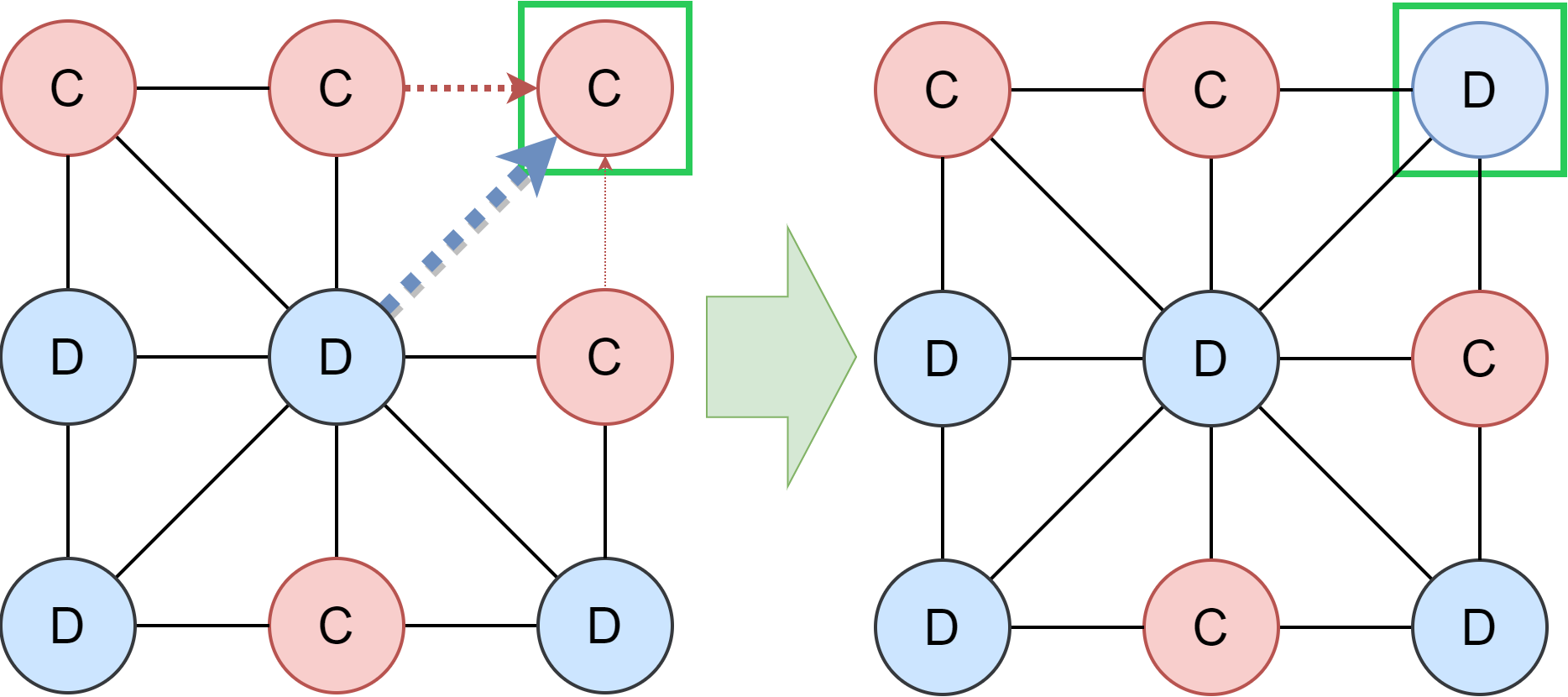}}
\subfigure[Strategy Update Order]{\includegraphics[scale=0.14]{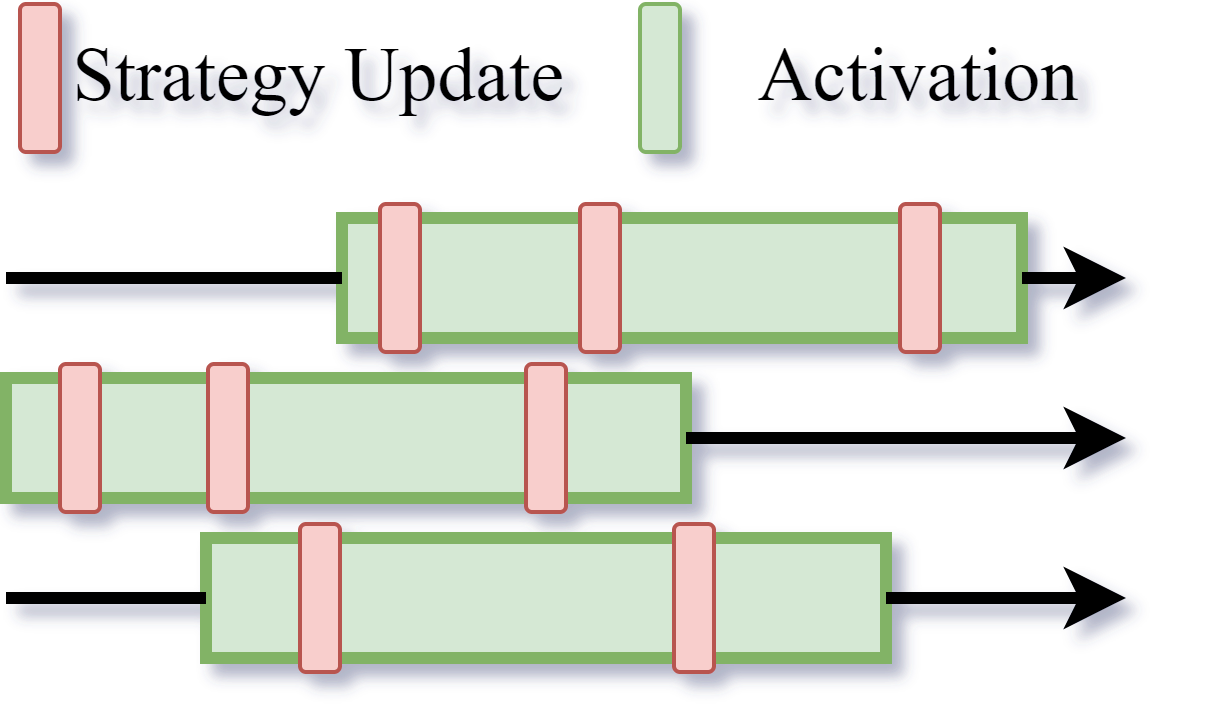}}
\vspace{-3.5mm}
\caption{\textbf{An example of the strategy update. } (a) The strategy update process of a $C$ player in the green square considers its neighbors. In a prisoner's dilemma game, the center vertex with the strategy $D$ possesses the highest payoff and fitness in the community because it connects to the largest number of cooperators. Therefore, the center vertex has the highest probability of spreading its strategy to the individual in the green square. (b) The time axis shows the strategy update timestamp of three individuals during activation. They update strategies by independent Poisson processes, thus only one vertex can be chosen to update at one certain moment. (Color online)}\label{fig: 2}
\end{figure}
\subsubsection{Theoretical Analysis of Evolutionary Game Dynamics}
Next, we provide some analysis of the evolutionary dynamic. In the following derivation, we prove the absorptivity of the population strategies and determine the critical cooperation condition.

\begin{proposition}\label{proposition: 4}
For an arbitrary two-strategy evolutionary game defined by Eq. \ref{eq: payoff} in the network $\mathcal{G}$ without mutation, the strategies of vertices in $\mathcal{G}$ reach the pure cooperative state or the pure defective state after sufficient time $t$ with power-law activating patterns. 
\end{proposition}
\begin{proof}
Let us define $q_{X|Y}$ as the conditional probability to find an $X$ player given that the adjacent individual plays $Y$ and consider the strategy update of a $D$ individual. Here, both $X$ and $Y$ stand for $C$ or $D$. Among the neighbors of the $D$ individual waiting for a strategy update, a $C$ player has the fitness
\begin{small}
\begin{equation}
f_C=1+\sum_{i=1}^{k}wP(k, i)\{(i-1)q_{C|C}R+[(i-1)q_{D|C}+1]S\}. 
\end{equation}
\end{small}
Similarly, a $D$ player has the fitness 
\begin{small}
\begin{equation}
f_D=1+\sum_{i=1}^{k}wP(k, i)\{(i-1)q_{C|D}T+[(i-1)q_{D|D}+1]P\}. 
\end{equation}
\end{small}
The term $P(k, i)$ is because of the Corollary \ref{corollary: 3}. The probability to find $k_C$ cooperators among the $D$ individual's $i$ activated neighbors is $\frac{i!}{(k_C!k_D!)}q_{C|D}^{k_C}q_{D|D}^{k_D}$. Based on the model description, the probability that a $D$ individual is influenced by its activated neighbors and becomes a $C$ individual with the probability $\frac{k_Cf_C}{k_Cf_C+k_Df_D}$. Therefore, using Taylor's formula, during a short time $\Delta t$, the frequency of cooperators increases by $1/N$ with the probability
\begin{small}
\begin{equation}\label{eq: pc incease}
\begin{aligned}
&P(\Delta p_C=\frac{1}{N})=\delta\Delta t\sum_{j=1}^{Np_D}P(Np_D, j)p_D\sum_{i=1}^{k}P(k, i)\\
&\times\sum_{k_C+k_D=i}\frac{i!}{(k_C!k_D!)}q_{C|D}^{k_C}q_{D|D}^{k_D}\frac{k_Cf_C}{k_Cf_C+k_Df_D}+o(\Delta t), 
\end{aligned}
\end{equation}
\end{small}
where $p_D$ is the frequency of the defectors. 

Then, we consider the update of a $C$ player. Correspondingly, its one $C$ neighbor has the fitness
\begin{small}
\begin{equation}
g_C=1+\sum_{i=1}^{k}wP(k, i)\{[(i-1)q_{C|C}+1]R+(i-1)q_{D|C}S\}, 
\end{equation}
\end{small}
and one $D$ neighbor has the fitness
\begin{small}
\begin{equation}
g_D=1+\sum_{i=1}^{k}wP(k, i)\{[(i-1)q_{C|D}+1]T+(i-1)q_{D|D}P\}. 
\end{equation}
\end{small}
The probability to find a configuration with $k_C$ cooperators in $i$ activated neighbors is $\frac{i!}{k_C!k_D!}q_{C|C}^{k_C}q_{D|C}^{k_D}$. According to our strategy updating rule, the probability that a given $C$ player studies the strategy $D$ is $\frac{k_Dg_D}{k_Cg_C+k_Dg_D}$. Accordingly, during $\Delta t$, the probability to find the frequency of cooperators decreases by $1/N$ is 
\begin{small}
\begin{equation}\label{eq: pc decrease}
\begin{aligned}
&P(\Delta p_C=-\frac{1}{N})=\delta\Delta t\sum_{j=1}^{Np_C}P(Np_C, j)p_C\sum_{i=1}^{k}P(k, i)\\
&\times\sum_{k_C+k_D=i}\frac{i!}{(k_C!k_D!)}q_{C|C}^{k_C}q_{D|C}^{k_D}\frac{k_Dg_D}{k_Cg_C+k_Dg_D}+o(\Delta t), 
\end{aligned}
\end{equation}
\end{small}
where $p_C$ is the frequency of cooperators. 

According to Eqs. \ref{eq: pc incease} and \ref{eq: pc decrease}, it is apparent that if $p_C=0$ or $p_D=0$, the probability to change the cooperation frequency is zero. Therefore, the pure cooperative and the pure defective states are both the only absorption states, while the other states are all transient ones. Besides, the transient states are visited finitely, which directly leads to the conclusion. 

Result follows. 
\end{proof}

Proposition \ref{proposition: 4} provides us with evidence of studying the fixation probability of different strategies because the absorption states are always reached. Before deriving the cooperation condition, we present the following lemma on a one-step random walk in the complex network with power-law activation patterns. 

\begin{lemma}\label{lemma: 3}
The expected one-step random walk probability from vertex $i$ to $j$ is
\begin{small}
\begin{equation}
l_{ij}=[1-(1-(\frac{(\mu-1)(\lambda-2)}{(\mu-1)(\lambda-2)+(\lambda-1)(\mu-2)})^{k_i}]/{k_i},
\end{equation}
\end{small}
where $k_i$ is the degree of vertex $i$. 
\end{lemma}
\begin{proof}
A random walk starts from $i$ and ends at the neighbor $j$ if and only if $j$ is activated. According to Lemma \ref{lemma: 2}, $j$ is activated with the expected probability $(\mu-1)(\lambda-2)/[(\mu-1)(\lambda-2)+(\lambda-1)(\mu-2)]$. If there are $d$ activated neighbors (including $j$) around $i$, the probability that the focal random walk steps into $j$ is $(1/d)$. Therefore, summing all the probability from $d=1$ to $k_i$, we have
\begin{small}
\begin{equation}
\begin{aligned}
l_{ij}&=\frac{E[N_1]}{N}\sum_{d=1}^{k_i}\frac{1}{d}C_{k_i-1}^{d-1}\times\\
&\frac{[(\mu-1)(\lambda-2)]^{d-1}[(\lambda-1)(\mu-2)]^{k_i-d}}{[(\mu-1)(\lambda-2)+(\lambda-1)(\mu-2)]^{k_i-1}}. \\
\end{aligned}
\end{equation}
\end{small}
By $(1/d)C_{k_i-1}^{d-1}=(1/k_i)C_{k_i}^{d}$, we have
\begin{small}
\begin{equation}
\begin{aligned}
l_{ij}&=\sum_{d=1}^{k_i}\frac{1}{k_i}C_{k_i}^{d}\frac{[(\mu-1)(\lambda-2)]^{d}[(\lambda-1)(\mu-2)]^{k_i-d}}{[(\mu-1)(\lambda-2)+(\lambda-1)(\mu-2)]^{k_i}}\\
&=[1-(1-(\frac{(\mu-1)(\lambda-2)}{(\mu-1)(\lambda-2)+(\lambda-1)(\mu-2)})^{k_i}]/{k_i}.\\
\end{aligned}
\end{equation}
\end{small}
Results follow. 
\end{proof}

We assume one single cooperator (\textit{resp. }defector) invades a network that all others are defectors (\textit{resp. }cooperator), and cooperation (\textit{resp. }defection) finally occupies the whole network with probability $\rho_C$ (\textit{resp. }$\rho_D$). If $\rho_C>\rho_D$, we say that cooperation is favored in this networked population. Next, we show the property of the critical cooperation condition in homogeneous graphs. 

\begin{proposition}\label{proposition: 5}
For the prisoner's dilemma game with $R=b-c$, $S=-c$, $T=b$, and $P=0$, $\rho_C>\rho_D$ if 
\begin{small}
\begin{equation}\label{eq: condition for cooperation}
(\frac{b}{c})>\frac{N-2}{N[1-(1-(\frac{(\mu-1)(\lambda-2)}{(\mu-1)(\lambda-2)+(\lambda-1)(\mu-2)})^{k}]/{k}-2}. 
\end{equation}
\end{small}
\end{proposition}
\begin{proof}
Suppose the networked evolutionary dynamic starts from the state $\mathbf{s_0}\in \{0,1\}^{\mathcal{G}}$. We focus on the proportion of cooperators at time $t$ as $\bar{S}(t)=\sum_{i\in \mathcal{G}}S_i(t)/N$. Following \cite{chen2013sharp} and \cite{allen2017evolutionary}, the fixation probability 
 of cooperation ($\rho_C$) is
 \begin{small}
\begin{equation}\label{eq: rhoc}
\rho_C=\frac{1}{N}+w\left \langle D' \right \rangle _{\mathbf{u}}^{\circ}+O(w^2),  
\end{equation}
\end{small}
where $D$ is a function of $\mathbf{s}$ that denotes the change of cooperation proportion, $\circ$ describes the limit under neutral drift ($w=0$), and $\mathbf{u}$ indicates that the initial state of $\mathbf{s}$ contains only one single cooperator selected from uniform distribution. The expected cooperation proportion change is
\begin{small}
\begin{equation}
\begin{aligned}
D(\mathbf{s})&=\sum_{i\in\mathcal{G}}\frac{s_i}{N}(\sum_{j\in\mathcal{G}}\frac{l_{ij}f_i(\mathbf{s})}{\sum_{k\in\mathcal{G}}l_{kj}f_k(\mathbf{s})}-1)\\
&=\frac{w}{N}\sum_{i\in\mathcal{G}}s_i(f_i^{(0)}(\mathbf{s})-f_i^{(2)}(\mathbf{s}))+O(w^2), 
\end{aligned}
\end{equation}
\end{small}
where $f_i^{(n)}$ denotes the expected payoff from $n$ steps away from $i$. Find the first derivative of $w$ and then we have
\begin{small}
\begin{equation}\label{eq: d'}
\begin{aligned}
D'(\mathbf{s})&=\frac{1}{N}\sum_{i\in\mathcal{G}}s_i(f_i^{(0)}(\mathbf{s})-f_i^{(2)}(\mathbf{s}))\\
&=\frac{1}{N}\sum_{i\in\mathcal{G}}s_i(-c(s_i^{(0)}-s_i^{(2)})+b(s_i^{(1)}-s_i^{(3)})). 
\end{aligned}
\end{equation}
\end{small}
Eq. \ref{eq: d'} can be calculated through coalescing random walk theory. Consider two random walkers starting from different vertices that step independently until their meeting, the time from the beginning to the end is the coalescence time. Denote the expected coalescence time from vertices $i$ and $j$ as $\tau_{ij}$. Taking expectation of $\tau_{ij}$ over all $n$-step random walk from $i$ to $j$, we have $\tau^{(n)}=\sum_{i,j\in\mathcal{G}}l_{ij}^{(n)}\tau_{ij}/N$, where $l_{ij}^{(n)}$ is the probability that an $n$-step random walk starts at $i$ and ends at $j$ in the network with power-law activation patterns. Results of spatial assortment condition \cite{allen2017evolutionary,allen2019mathematical} suggest that 
\begin{small}
\begin{equation}\label{eq: assortment condition}
\left \langle \frac{1}{N}\sum_{i\in\mathcal{G}}s_i (s_i^{(n_1)}-s_i^{(n_2)}) \right \rangle _{\mathbf{u}}^{\circ}=\frac{\tau^{(n_2)}-\tau^{(n_1)}}{2N}. 
\end{equation}
\end{small}
Since the coalescence time has the recurrence relation
\begin{small}
\begin{equation}
\tau_{ij}=\left\{
\begin{aligned}
&0 ,  &i=j \\
&1+\frac{1}{2}\sum_{k\in\mathcal{G}}(l_{ik}\tau_{jk}+l_{jk}\tau_{ik}) , &i\neq j,
\end{aligned}
\right.
\end{equation}
\end{small}
we have 
\begin{small}
\begin{equation}
\begin{aligned}
&\tau^{(n)}=\frac{1}{N}\sum_{i,j\in\mathcal{G}}l_{ij}^{(n)}+\frac{1}{2N}[\sum_{i,j,k\in\mathcal{G}}l_{ji}^{(n)}l_{ik}\tau_{jk}+\sum_{i,j,k\in\mathcal{G}}l_{ij}^{(n)}l_{jk}\tau_{ik}]\\
&-\frac{1}{N}\sum_{i\in\mathcal{G}}l_{ii}^{(n)}(1+\sum_{k\in\mathcal{G}}l_{ik}\tau{ik})\\
&=\frac{1}{N}\frac{(\mu-1)(\lambda-2)}{(\mu-1)(\lambda-2)+(\lambda-1)(\mu-2)}N+\\
&\frac{1}{2N}[\sum_{j,k\in\mathcal{G}}l_{jk}^{(n+1)}\tau_{jk}+\sum_{i,k\in\mathcal{G}}l_{ik}^{(n+1)}\tau_{ik}]-\frac{1}{N}\sum_{i\in\mathcal{G}}l_{ii}^{(n)}(1+\sum_{k\in\mathcal{G}}l_{ik}\tau_{ik}).\\
\end{aligned}
\end{equation}
\end{small}
Rewrite $\tau_i=1+\sum_{k\in\mathcal{G}}l_{ik}\tau_{ik}$ and we obtain the recurrence
\begin{small}
\begin{equation}\label{eq: rec}
\begin{aligned}
\tau^{(n+1)}&=\tau^{(n)}+\frac{1}{N}\sum_{i\in\mathcal{G}}l_{ii}^{(n)}\tau_i-\frac{(\mu-1)(\lambda-2)}{(\mu-1)(\lambda-2)+(\lambda-1)(\mu-2)}. 
\end{aligned}
\end{equation}
\end{small}
Note that if $n\rightarrow\infty$ in Eq. \ref{eq: rec}, we have 
\begin{small}
\begin{equation}
\frac{1}{N^2}\sum_{i\in\mathcal{G}}\tau_i=\frac{(\mu-1)(\lambda-2)}{(\mu-1)(\lambda-2)+(\lambda-1)(\mu-2)}. 
\end{equation}
\end{small}
Additionally, the recurrence relation shows that $\tau^{(0)}=0$, and 
\begin{small}
\begin{equation}\label{eq: tau1}
\tau^{(1)}=\frac{1}{N}\sum_{i\in\mathcal{G}}\tau_i-\frac{(\mu-1)(\lambda-2)}{(\mu-1)(\lambda-2)+(\lambda-1)(\mu-2)},
\end{equation}
\end{small}
\begin{small}
\begin{equation}\label{eq: tau2}
\tau^{(2)}=\frac{1}{N}\sum_{i\in\mathcal{G}}\tau_i-\frac{2(\mu-1)(\lambda-2)}{(\mu-1)(\lambda-2)+(\lambda-1)(\mu-2)},
\end{equation}
\end{small}
\begin{small}
\begin{equation}\label{eq: tau3}
\tau^{(3)}=\frac{1}{N}\sum_{i\in\mathcal{G}}\tau_i(1+l_{ii}^{(2)})-\frac{3(\mu-1)(\lambda-2)}{(\mu-1)(\lambda-2)+(\lambda-1)(\mu-2)}, 
\end{equation}
\end{small}
where we drop $l_{ii}^{(1)}=0$ since we do not consider self-loop. Combining Eqs. \ref{eq: rhoc}, \ref{eq: assortment condition}, and \ref{eq: tau1}-\ref{eq: tau3}, we get the fixation probability of cooperation under weak selection
\begin{small}
\begin{equation}\label{eq: fixation probability1}
\begin{aligned}
&\rho_C=\frac{1}{N}+\frac{w}{2N}\times\\
&(-c(\frac{1}{N}\sum_{i\in\mathcal{G}}\tau_i-\frac{2(\mu-1)(\lambda-2)}{(\mu-1)(\lambda-2)+(\lambda-1)(\mu-2)})+\\
&b(\frac{1}{N}\sum_{i\in\mathcal{G}}\tau_i l_{ii}^{(2)}-\frac{2(\mu-1)(\lambda-2)}{(\mu-1)(\lambda-2)+(\lambda-1)(\mu-2)}))+O(w^2). 
\end{aligned}
\end{equation}
\end{small}
Reversing the payoff matrix directly leads to the fixation probability of defection
\begin{small}
\begin{equation}\label{eq: fixation probability2}
\begin{aligned}
&\rho_D=\frac{1}{N}+\frac{w}{2N}\times\\
&(c(\frac{1}{N}\sum_{i\in\mathcal{G}}\tau_i-\frac{2(\mu-1)(\lambda-2)}{(\mu-1)(\lambda-2)+(\lambda-1)(\mu-2)})-\\
&b(\frac{1}{N}\sum_{i\in\mathcal{G}}\tau_i l_{ii}^{(2)}-\frac{2(\mu-1)(\lambda-2)}{(\mu-1)(\lambda-2)+(\lambda-1)(\mu-2)}))+O(w^2). 
\end{aligned}
\end{equation}
\end{small}

Therefore, based on Lemma \ref{lemma: 3}, Eqs. \ref{eq: fixation probability1} and \ref{eq: fixation probability2}, with the fact that $\rho_C>\rho_D \iff \rho_C>\frac{1}{N}>\rho_D$, we obtain the condition as Eq. \ref{eq: condition for cooperation}.

Result follows. 
\end{proof}

In Proposition \ref{proposition: 5}, we present the cooperation condition for the homogeneous network with degree $k$. This ends the model section. We have introduced the network model with individuals' power-law activating patterns and the evolutionary game dynamics, which have been theoretically analyzed via stochastic methods. 
\section{Simulation}\label{Sec: III}
\small
This section first illustrates technical details in our simulation and then shows the results to verify the proposed theorems. We first introduce our experimental methods in Sec. \ref{Sec: III(A)}. In Sec. \ref{Sec: III(B)}, we validate Proposition \ref{proposition: 1} by giving the statistical distribution of activated network sizes. In Sec. \ref{Sec: III(C)}, we provide the activated subgraph degree distributions and mean degrees of WSNs and BANs. In Sec. \ref{Sec: III(D)}, we analyze the resilience and robustness of the activated subgraph caused by the intermittent interaction. In Secs. \ref{Sec: III(E)} and \ref{Sec: III(F)}, we verify Propositions \ref{proposition: 4} and \ref{proposition: 5} by providing the results of fixation probability on RRGs and WSNs. In Sec. \ref{Sec: III(G)}, we present and discuss the results of evolutionary dynamics with our proposed power-law activation in real-world network data sets. 
\subsection{Methods}\label{Sec: III(A)}
\small
\subsubsection{Power-law activating patterns}
As mentioned in Sec. \ref{sec: II(A)}, an individual stays activated for a random time that follows the power-law distribution with the parameter $\mu$, and stays quiescent for a random time that follows the power-law distribution with the parameter $\lambda$. To simulate this on continuous time stamps, we consider the transformation method as follows \cite{press2007numerical} \cite{clauset2009power}. The complementary cumulative distribution function of a power-law distribution denoted in the same form as Eq. \ref{eq: power-law distribution} is 
\begin{small}
\begin{equation}
F(t;\alpha)=Pr(T\geq t)=\int_{t}^{\infty}f(t'; \alpha)\mathrm{d}t'=(\frac{t}{t_0})^{1-\alpha}. 
\end{equation}
\end{small}
Denote $r\sim U(0, 1)$ as a random number that follows a uniform distribution. According to the mentioned transformation method, a variable that follows a power-law distribution with the parameter $\alpha$ is
\begin{small}
\begin{equation}\label{eq: power-law random}
t=F^{-1}(1-r; \alpha)=t_0 (1-r)^{\frac{1}{1-\alpha}}=t_0 r^{\frac{1}{1-\alpha}}. 
\end{equation}
\end{small}
The last equality holds because when $r\sim U(0, 1)$, we have $1-r\sim U(0,1)$ as well. 

It is worth noting that the above method may generate very large numbers which affect the simulation of dynamic processes. We here set an upper bound of the power-law random number, say $\Gamma$, to avoid the overtime issue. We determine the value ranges of the power-law parameters ($\lambda$ and $\mu$) and $\Gamma$ to minimize the simulation error. The expectation of a power-law random number in the range $[t_0, \Gamma]$ is 
\begin{small}
\begin{equation}
E[t]=\int_{t_0}^{\Gamma}tf(t;\alpha)\mathrm{d}t=\frac{\alpha-1}{2-\alpha}t_0^{\alpha-1}\Gamma^{2-\alpha}+\frac{\alpha-1}{\alpha-2}t_0. 
\end{equation}
\end{small}
We set the upper bound $\Gamma=10^4$ and the power-law parameter $\alpha>2.5$ to ensure $\frac{\alpha-1}{2-\alpha}t_0^{\alpha-1}\Gamma^{2-\alpha}\rightarrow 0$, leading to acceptable errors. 

In our simulations, each individual becomes activated or quiescent initially with equal probability ($50\%$). If a vertex is activated in a moment, we generate a random time based on Eq. \ref{eq: power-law random} with $\alpha=\mu$ to keep its state, and then becomes quiescent. Similarly, if a vertex becomes quiescent, a random time with $\alpha=\lambda$ is generated for the quiescent state and then the vertex becomes activated. 
\subsubsection{Poisson Process}
We assume each activated vertex updates the strategy by a Poisson process with the parameter $\delta=1$. Accordingly, the inter-event time of the strategy update occurrence follows an exponential distribution with the parameter $\delta=1$. Once a vertex becomes activated from quiescent, we generate a set of exponentially distributed time stamps as the strategy update time until it becomes quiescent again. 
\subsubsection{Network Structure}
We consider three types of generative networks. \textit{(a) Random regular graphs (RRG)}: each vertex is connected to $k$ neighbors randomly.  \textit{(b) Watts-Strogatz small-world network (WSN)}: WSN is constructed by random reconnection (here with probability $p=0.40$) of a nearest-neighbor couple network with the mean degree $k$ \cite{watts1998collective}. \textit{(c) Barab\'{a}si-Albert scale-free network (BAN): }BAN is constructed by growth and preferential attachment, i.e., each new vertex connects to $m$ existing nodes with the probability proportion to degrees and yields the mean degree as $k=2m$ \cite{barabasi1999emergence}. 
\subsubsection{Experimental Environment}
We establish all the following experiments based on \textit{Python 3.8}. We employ the package \textit{networkx} to generate a network structure and \textit{numpy.random} to generate the random time sequences. 
\begin{figure*}                                              %fig3
\centering
\subfigure[$\lambda=2.60$]{\includegraphics[scale=0.26]{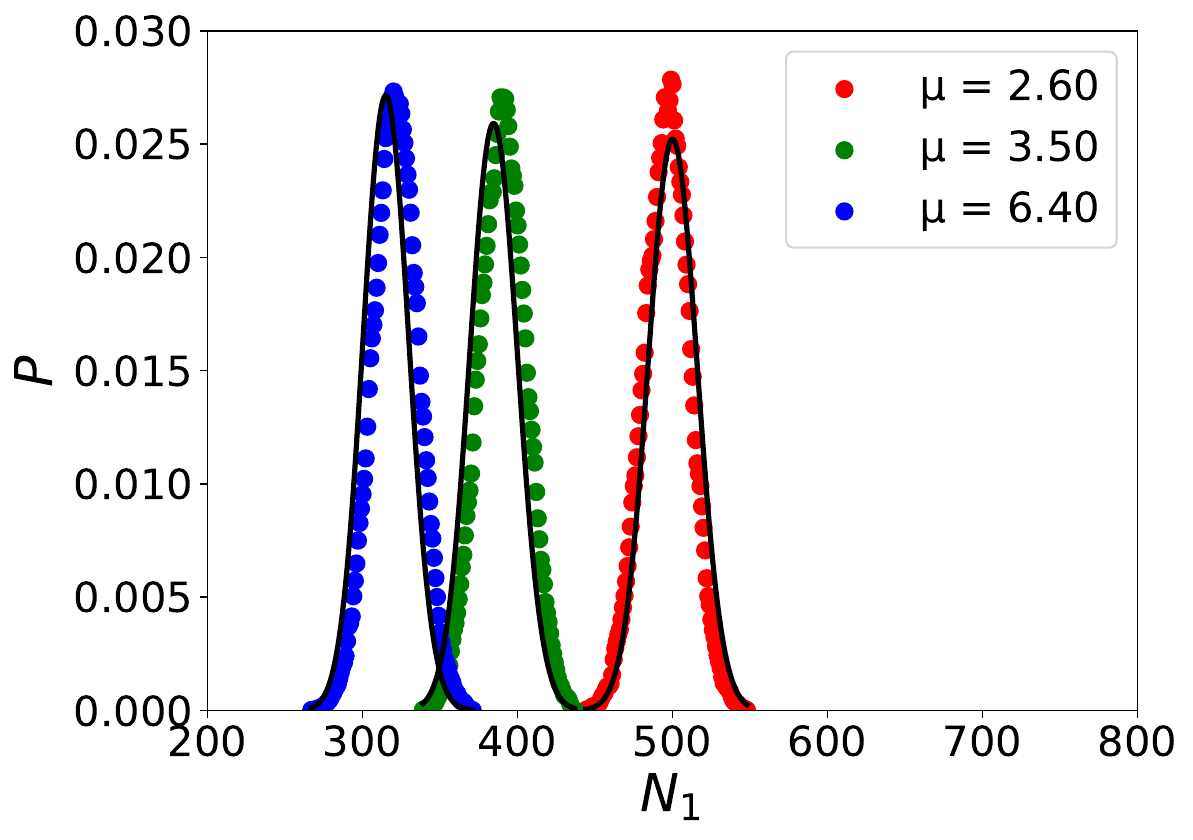}}
\subfigure[$\lambda=3.50$]{\includegraphics[scale=0.26]{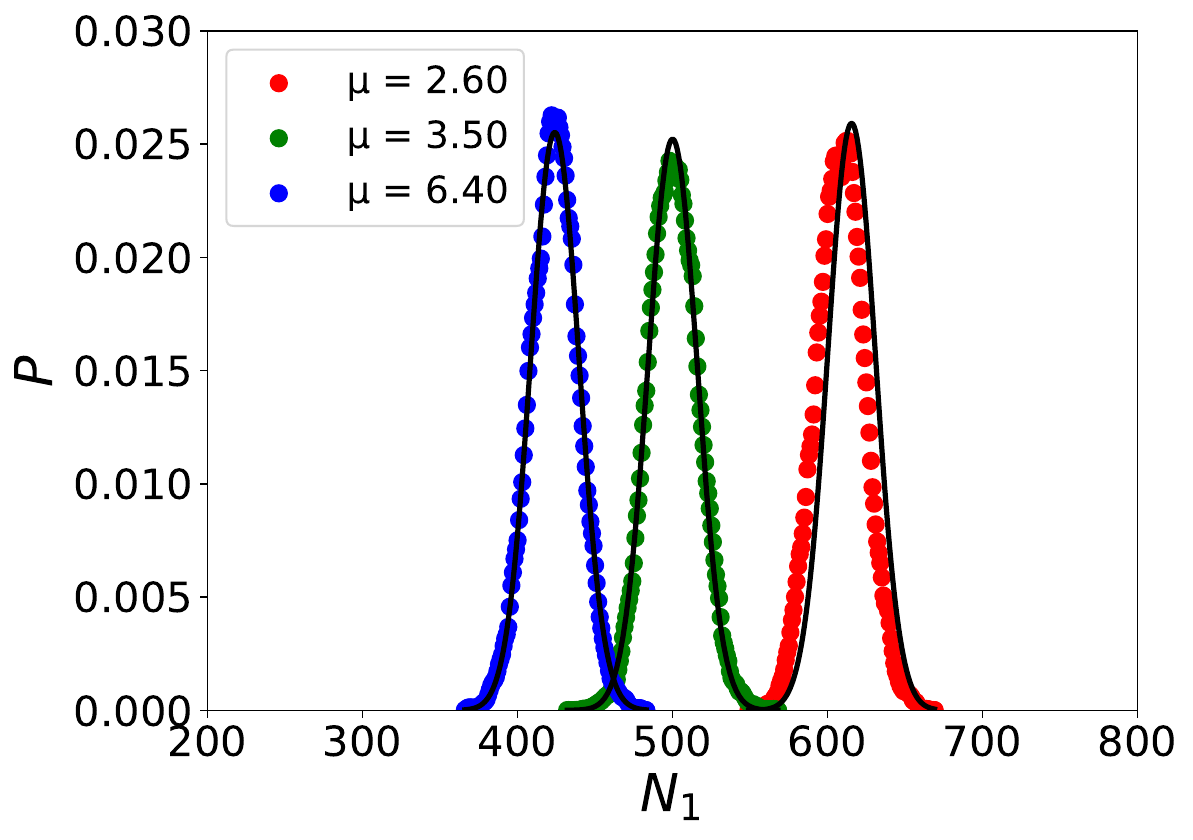}}
\subfigure[$\lambda=6.40$]{\includegraphics[scale=0.26]{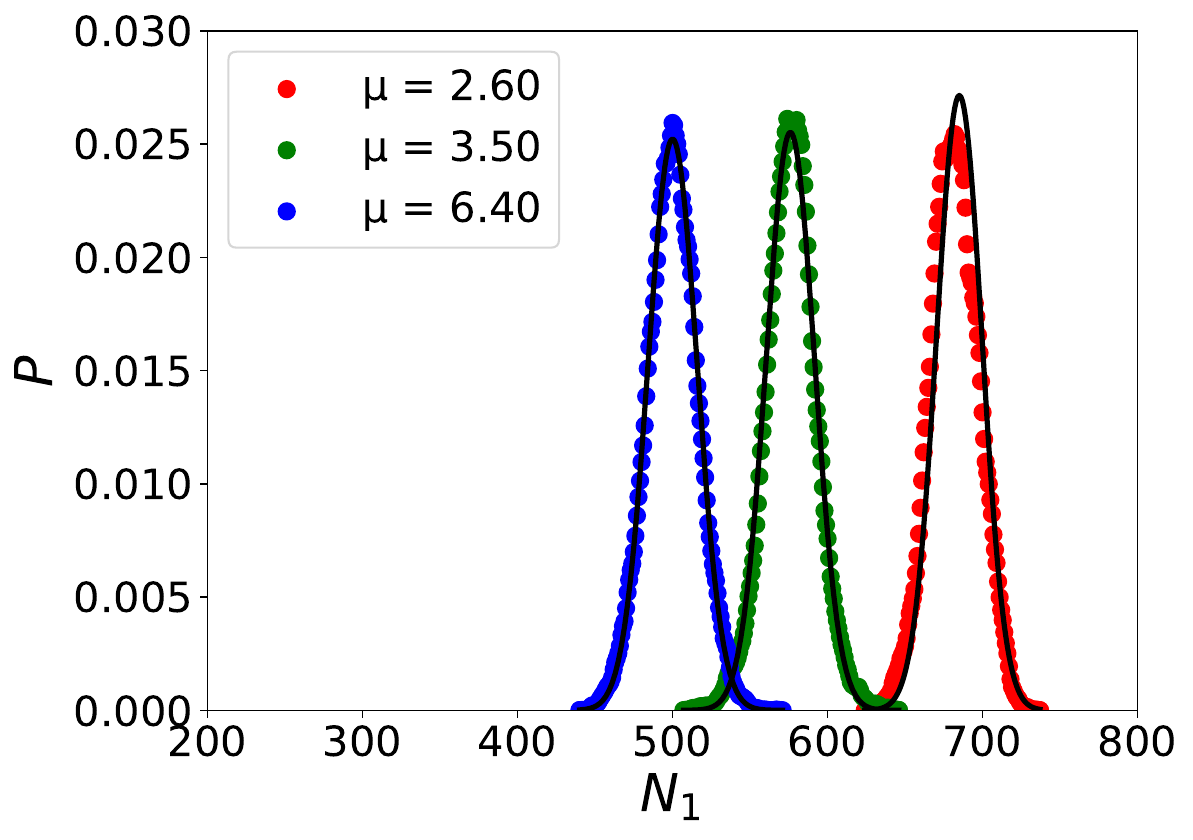}}
\vspace{-3.5mm}
\caption{\textbf{Size distributions of the activated subgraph. }The activated subgraph size shows binomial distribution as Proposition \ref{proposition: 1} suggests. (a) $\lambda=2.60$. (b) $\lambda=3.50$. (c) $\lambda=6.40$. We set $\mu=[2.60, 3.50, 6.40]$ and network size $N=1000$ for cross simulation. Initially, each vertex has the equal probability to be activated or quiescent. We collect the sizes of activated subgraphs if $t>50$ and calculate the frequencies until $t=600$. Results for $\mu=2.60$, $3.50$, and $6.40$ are shown in red, green, and blue respectively. Black plots indicate the theoretical distribution as Proposition \ref{proposition: 1}. (Color online)}\label{fig: 3}
\end{figure*}

\subsection{Statistical Characteristics of Activated Subgraph Size}\label{Sec: III(B)}

We first conduct simulations on the sizes of the activated subgraph ($\mathcal{G}_1(t)$) to validate Proposition \ref{proposition: 1}. Notably, the scale of the activated subgraph is independent of network type; therefore, this subsection does not treat network type as a variable. Fig. \ref{fig: 3} shows the frequencies of activated vertex numbers for an initial network size of $N = 1000$. We set the power-law rates as $[2.60, 3.50, 6.40]$ for cross-simulation. The red, green, and blue points represent the experimental size distributions for $\mu = [2.60, 3.50, 6.40]$, respectively, while the black solid lines correspond to the theoretical results derived from Eq. \ref{eq: proposition 1}. The close agreement between the theoretical and experimental results confirms the accuracy of Proposition \ref{proposition: 1} in describing real systems. The results in Fig. \ref{fig: 3} suggest that during the evolution process, the size of the activated subgraph follows a consistent probability distribution despite heterogeneous activation patterns. Specifically, if $\lambda$ is fixed, a smaller $\mu$ leads to a higher size expectation. Similarly, if $\mu$ is fixed, a larger $\lambda$ ensures a greater activated subgraph size.

\begin{table}[h]
\centering
\caption{\textbf{The Kullback-Leibler divergence of the theoretical and experimental distributions.}}
\begin{tabular}{c|c|c|c}
\toprule[2pt]
\hline
\diagbox{$\lambda$}{$\mu$}&2.60      &3.50        &6.40\\
\midrule
2.60&0.005&0.066&0.036\\
3.50&0.092&0.002&0.06\\
6.40&0.049&0.003&0.003\\
\midrule
\bottomrule[2pt]
\end{tabular}
\label{tab: 1}
\vspace{-4mm}
\end{table}
\begin{figure}[htbp]                                              %fig4
\centering
\subfigure[WSN, $\mu=2.60$]{\label{fig: 4(a)}
\includegraphics[scale=0.22]{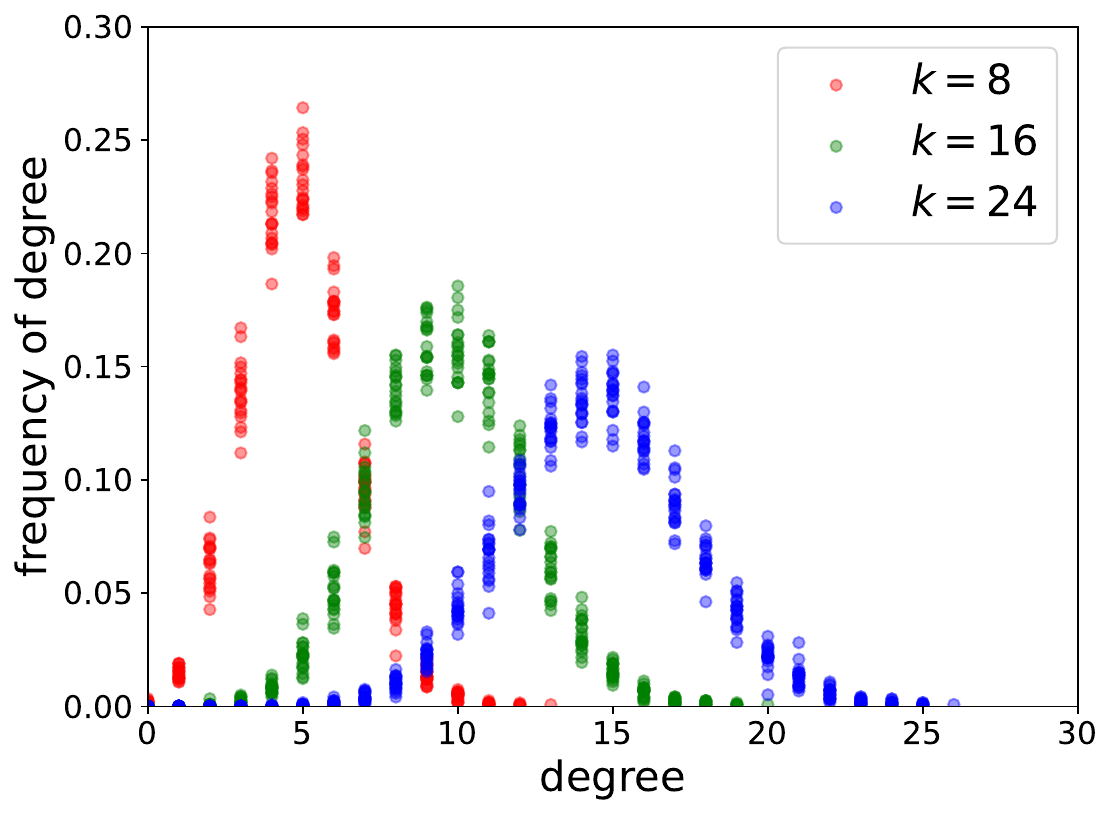}}
\subfigure[WSN, $\mu=3.70$]{\label{fig: 4(b)}
\includegraphics[scale=0.22]{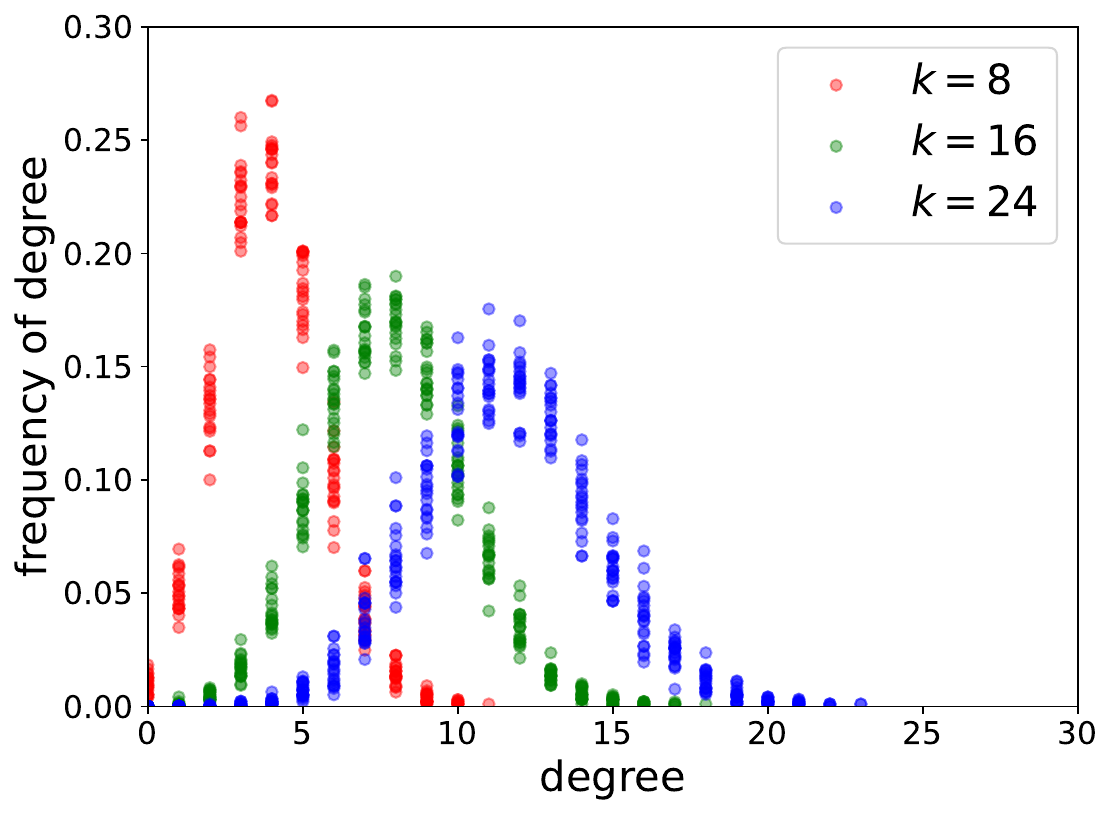}}
\vspace{-3.5mm}

\subfigure[BAN, $\mu=2.60$]{\label{fig: 4(c)}
\includegraphics[scale=0.22]{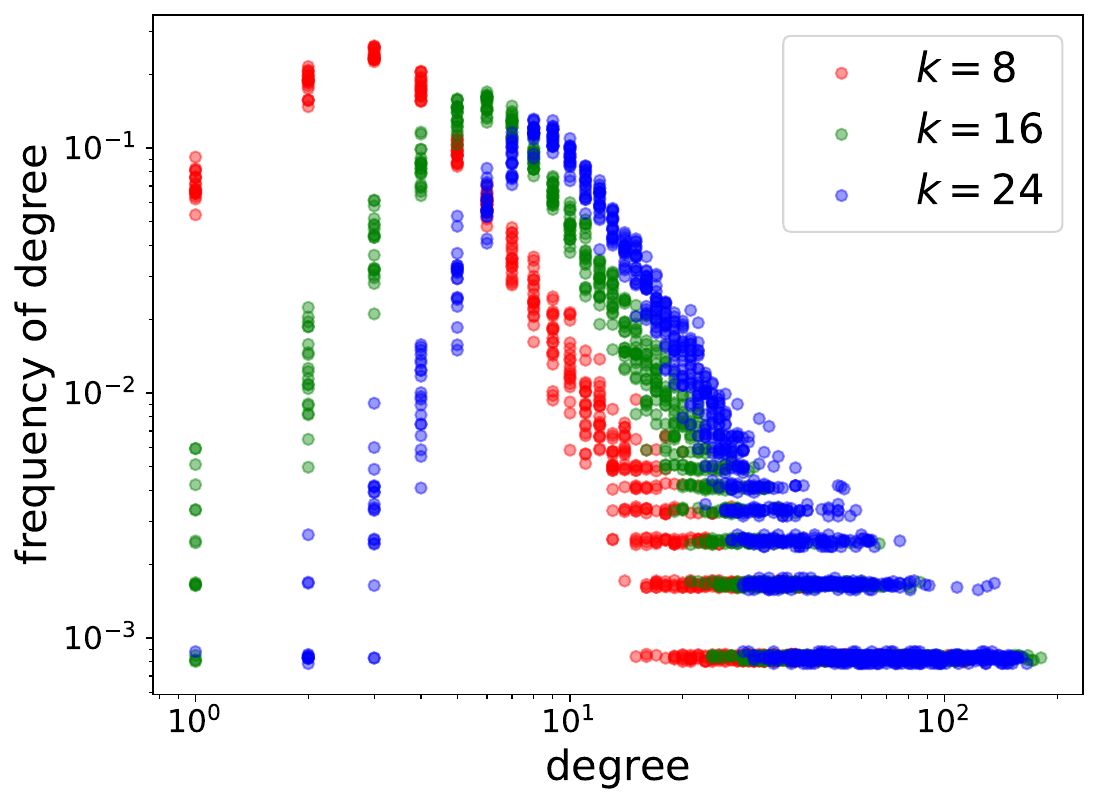}}
\subfigure[BAN, $\mu=3.70$]{\label{fig: 4(d)}
\includegraphics[scale=0.22]{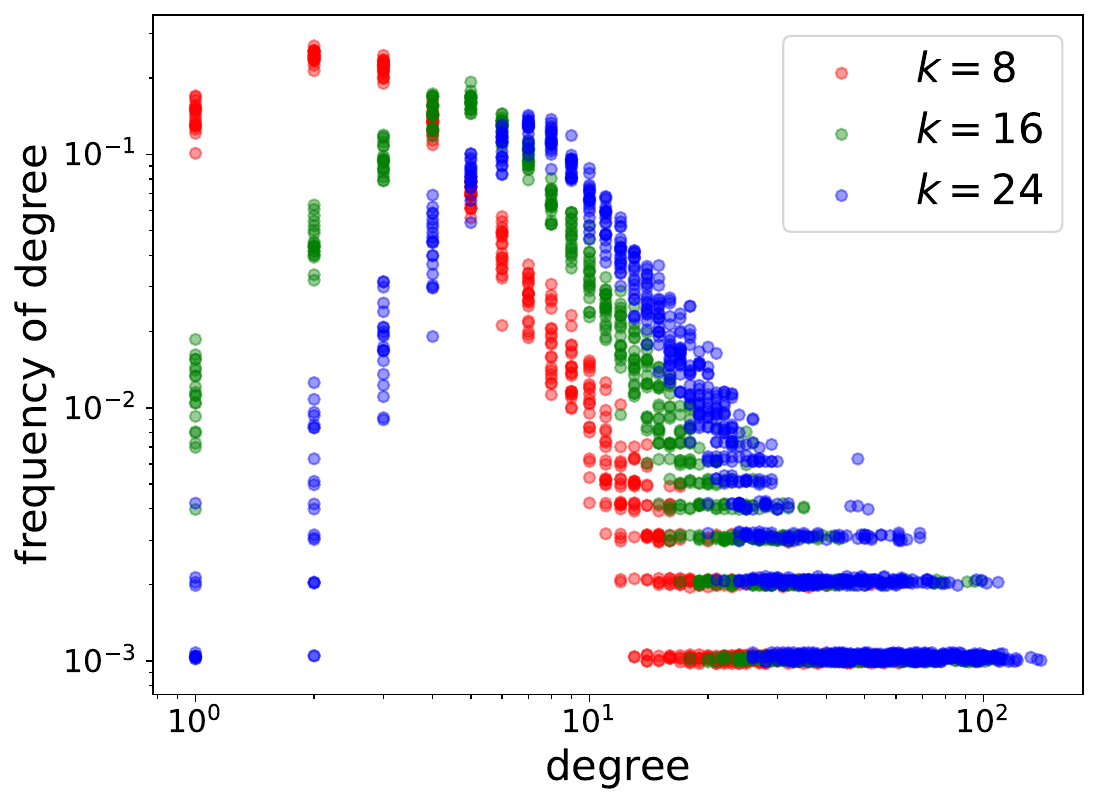}}
\vspace{-3.5mm}
\caption{\textbf{Degree distributions of the activated subgraph. } The power-law activation patterns maintain the homogeneity of WSN but break the heterogeneity of BAN. (a) WSN, $\mu=2.60$. (b) WSN, $\mu=3.70$. (c) BAN, $\mu=2.60$. (d) BAN, $\mu=3.70$. We fix the parameters $N=2000$, $\lambda=3.50$ and set $\mu=[2.60, 3.70]$, $k=[8, 16, 24]$ for cross experiments. We collect the degree distributions if $t>20$ every $10$ simulation time until $t=200$. Results for $k=8$, $16$, and $24$ are shown in red, green, and blue respectively. Degree distributions of WSNs and BANs are shown in double-linear axes and double-logarithmic axes respectively. (Color online)}\label{fig: 4}
\end{figure}
To further demonstrate the effectiveness of the mentioned proposition, we here employ the Kullback-Leibler (KL) divergence to describe the distances between two probability distributions, denoted as $KL_{Q\vert\vert P}=\sum_{i}Q_i \log \frac{Q_i}{P_i}$, where $P_i$ is denoted in Eq. \ref{eq: proposition 1} as the theoretical results, and $Q_i$ is the experimental frequency of activated subgraph. Calculating the KL divergence directly may exceed the computer's upper limit. Therefore, we transfer Eq. \ref{eq: proposition 1} and calculate $P_i$ after logarithmic mapping by 
\begin{small}
\begin{equation}\label{eq: KL divergence transformed}
\begin{aligned}
&\log{P_i}=\log{C_{N}^{i}}+i \log{(\mu-1)(\lambda-2)}+(N-i)\\
&\log{(\lambda-1)(\mu-1)})-N\log{(\mu-1)(\lambda-2)+(\lambda-1)(\mu-2)}. 
\end{aligned}
\end{equation}
\end{small}
In Tab. \ref{tab: 1}, we present values of the KL divergence between theoretical and experimental distributions. For all parameter sets, the results are close to zero, which indicates small distances. Accordingly, we can conclude that Proposition \ref{proposition: 1} is an accurate theoretical result. We note that the KL divergence results are obtained with the fact that our algorithm to generate power-law random time is biased to avoid the low probability overflow as mentioned above in Sec. \ref{Sec: III(A)}. 

\subsection{Network Topology of Activated Subgraph}\label{Sec: III(C)}

We now focus on the network topology properties of the activated subgraph. We employ two of the most representative network models as introduced in Sec. \ref{Sec: III(A)}, including the WSNs and BANs. As shown in Fig. \ref{fig: 4}, we obtain several groups of degree distributions of activated subgraph for different $k$ ($[8, 16, 24]$) and $\mu$ ($[2.60, 3.70]$) with fixed $\lambda=3.50$ and $N=2000$. The degree frequency functions of activated subgraphs are shown in $t\in (20, 200)$ with the step $\Delta t=10$. 
\begin{figure}                                              %fig5
\centering
\subfigure[WSN]{\label{fig: 5(a)}
\includegraphics[scale=0.21]{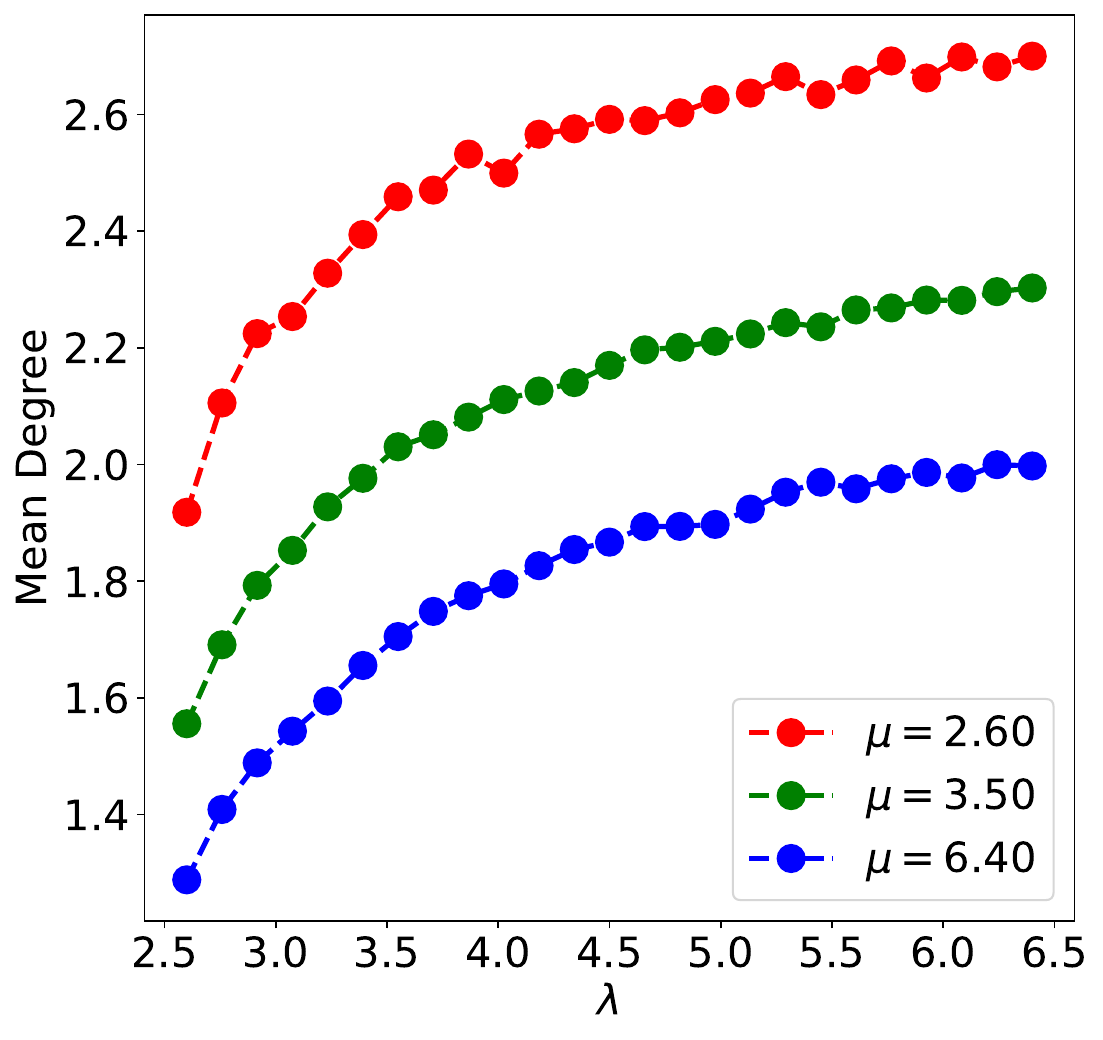}}
\subfigure[BAN]{\label{fig: 5(b)}
\includegraphics[scale=0.21]{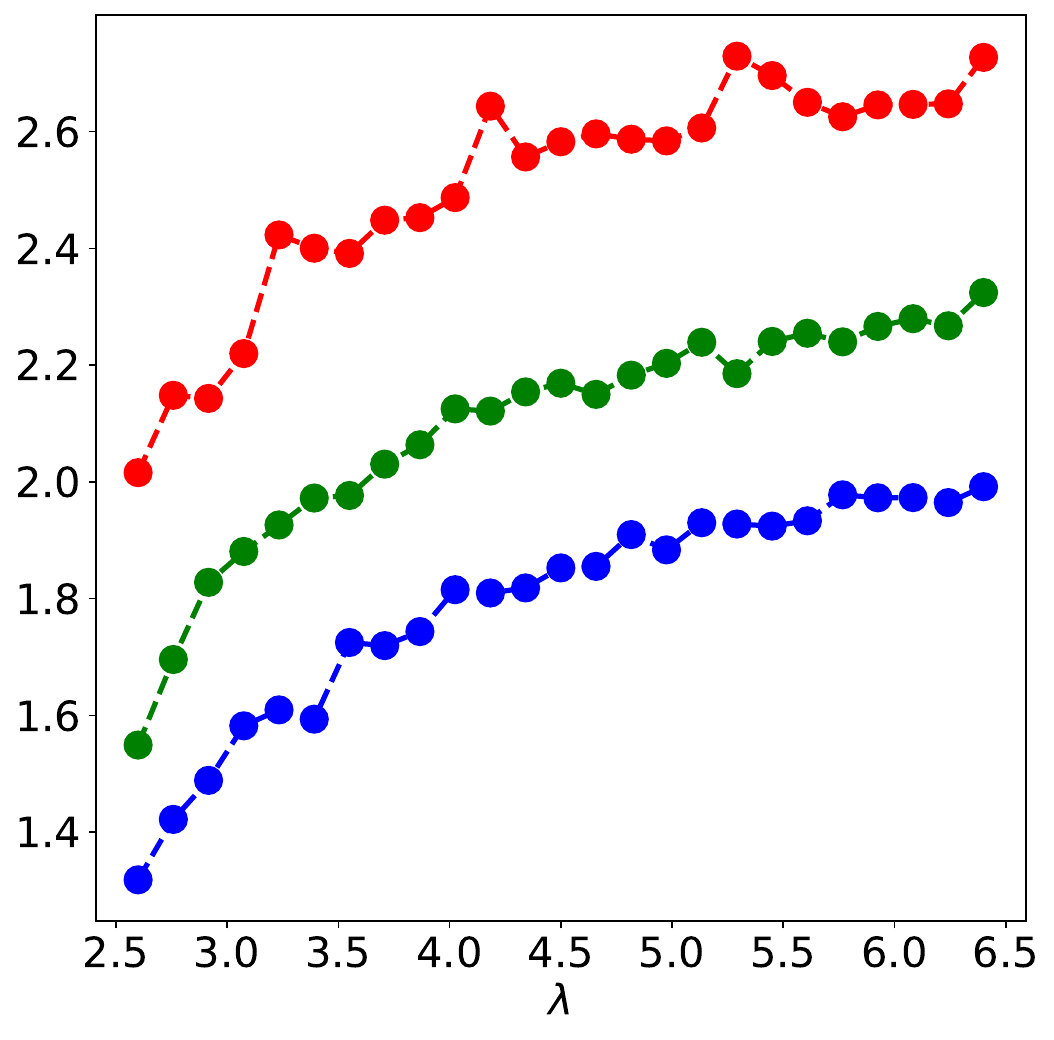}}
\vspace{-3.5mm}

\subfigure[Mouse Visual Cortex]{\label{fig: 5(c)}
\includegraphics[scale=0.21]{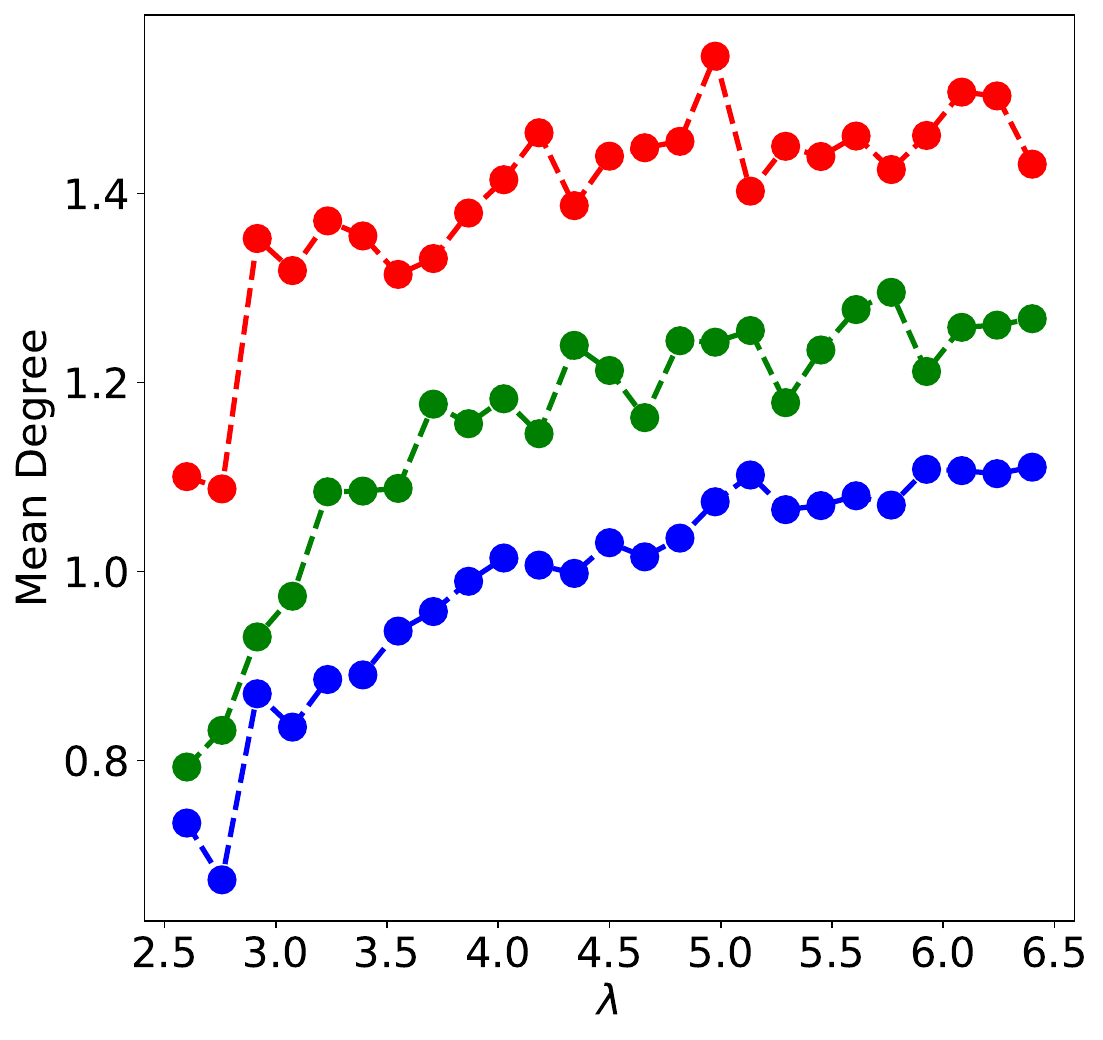}}
\subfigure[Infect Dublin]{\label{fig: 5(d)}
\includegraphics[scale=0.21]{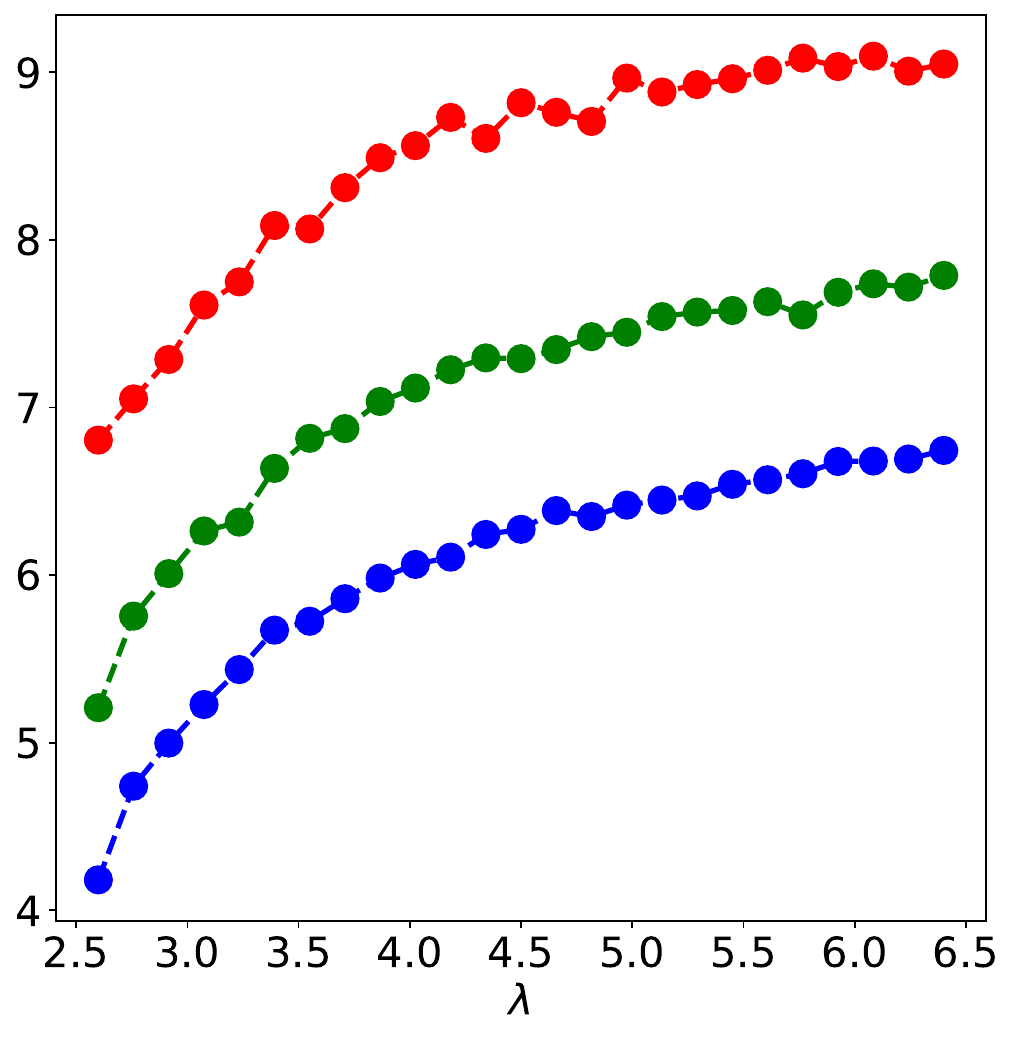}}
\vspace{-3.5mm}
\caption{\textbf{Mean degree as the function of $\lambda$. }Mean degrees of activated subgraphs have positive correlations to $\lambda$. (a) WSN, $N=500$, $k=4$. (b) BAN, $N=500$, $k=4$. (c) Mouse visual cortex network, $N=193$. (d) Infect Dublin, infectious contact network, $N=410$. We set $\lambda\in[2.60,6.40]$ and $\mu\in\{2.60, 3.50, 6.40\}$ for cross simulations. We collect the mean degrees of activated subgraphs in the time interval $t\in[50,150]$ and take the average for each data point. The results for $\mu=2.60, 3.50, 6.40$ are shown in red, green, and blue curves respectively.  (Color online)}\label{fig: 5}
\end{figure}

The degree distribution results are shown to remain stable over time. Figs \ref{fig: 4(a)} and \ref{fig: 4(b)} present the results based on WSNs using double-linear axes. In a WSN, the degree distribution is homogeneous with an expected value of $k$. Under the proposed power-law activation patterns for vertices, the degree distributions of the activated subgraph remain homogeneous but exhibit a smaller expectation. This phenomenon aligns with the theoretical results in Corollary \ref{corollary: 3}. Figs \ref{fig: 4(c)} and \ref{fig: 4(d)} display the degree frequencies of activated subgraphs for BANs using double-logarithmic axes to facilitate comparison with linear axes. The degree distributions of BANs appear as a straight line in double-logarithmic coordinates, indicating the heavy-tailed nature of these networks. However, the vertex power-law activation patterns reduce this heavy-tail property in the activated subgraph, resulting in a more homogeneous structure. Specifically, the degree distributions show an increase in low-degree vertices, while the heavy tail persists for high-degree vertices. Moreover, the probability of encountering a low-degree vertex is higher when $k$ is smaller. A plausible explanation for the reduced heterogeneity in BANs is that high-degree vertices in the underlying network tend to have many quiescent neighbors, which decreases their degree in the activated subgraph $\mathcal{G}_1(t)$. Additionally, since the minimum degree in a BAN corresponds to the new connection number $m$, several vertices with degree $m$ have independently and identically distributed activated neighbors, as indicated by Corollary \ref{corollary: 3}.

\begin{table}[h]
\centering
\caption{\textbf{The skewness and kurtosis of degree distributions. }}\label{tab: 2}
\vspace{-3mm}
\begin{tabular}{c|ccc|ccc}
\toprule[2pt]
\midrule
 & \multicolumn{3}{c|}{Skewness}                                                                        & \multicolumn{3}{c}{Kurtosis}                                                                        \\ \midrule
$k$     & \multicolumn{1}{c|}{$8$} & \multicolumn{1}{c|}{$16$} & $24$ & \multicolumn{1}{c|}{$8$} & \multicolumn{1}{c|}{$16$} & $24$ \\ \midrule
WSN, $\mu=2.60$                      & \multicolumn{1}{c}{0.237}                     & \multicolumn{1}{c}{0.152}                    & 0.097                     & \multicolumn{1}{c}{-0.008}                     & \multicolumn{1}{c}{0.032}                    & 0.005                    \\ 
WSN, $\mu=3.70$                     & \multicolumn{1}{c}{0.248}                     & \multicolumn{1}{c}{0.179}                    & 0.168                     & \multicolumn{1}{c}{0.005}                     & \multicolumn{1}{c}{-0.011}                    & -0.071                     \\ 
BAN, $\mu=2.60$                     & \multicolumn{1}{c}{6.899}                     & \multicolumn{1}{c}{5.768}                    & 4.876                     & \multicolumn{1}{c}{71.449}                     & \multicolumn{1}{c}{40.848}                    & 36.489                     \\
BAN, $\mu=3.70$                     & \multicolumn{1}{c}{5.952}                     & \multicolumn{1}{c}{5.010}                    & 4.193                     & \multicolumn{1}{c}{59.837}                     & \multicolumn{1}{c}{42.040}                    & 32.390                     \\
\midrule
\bottomrule[2pt]
\end{tabular}
\vspace{-4mm}
\end{table}
To further study the characteristics of the degree distributions, below we present the statistics of the degree distributions in Tab. \ref{tab: 2}. The skewness is the third-order normalized moment that describes the deviation direction and degree. The kurtosis is employed to describe the peak height at the average value of the curve, where $3.0$ is subtracted from the result to get $0.0$ for a normal distribution. As shown in the results for WSNs, the skewness and kurtosis are all close to zero, indicating a symmetric degree distribution and moderate kurtosis. Based on this, we conclude that the power-law activating patterns of vertices maintain the homogeneous property of the WSNs in the activated subgraph. In addition, the skewness of the activated subgraph degree distribution in WSNs decreases as the increase of $k$ and shows little right deviation. However, for BANs, the skewness and kurtosis are all greater than zero. The skewness decreases as $k$ increases and shows a stronger right deviation than the WSNs. Results in Fig. \ref{fig: 4} present that the activated subgraph of a BAN with power-law activating patterns shows a peak in degree distributions. Combined with the obtained kurtosis, the degree distributions are high and pointed. Accordingly, we conclude that the heterogeneity is split in BANs with the power-law activating patterns. 

The power-law activating patterns can affect the number of activated vertices around a single node, changing the mean degree of the activated subgraph. In Fig. \ref{fig: 5}, we present the mean degrees of activated subgraphs for four different networks, including two real network data sets.  Here, we discuss $\lambda\in[2.60, 6.40]$ and $\mu\in\{2.60, 3.50, 6.40\}$ for cross simulations. We note that the mean degrees of the activated subgraphs increase as the activation rates $\lambda$s grow. Additionally, a small quiescent rate $\mu$ ensures the large mean degree of the activated subgraph. This shows similar properties compared to the conclusion on the activated network sizes. The used real network data sets in Figs. \ref{fig: 5(c)} and \ref{fig: 5(d)} are from \cite{infect,bigbrain,nr-aaai15}. 

\subsection{Network Resilience and Robustness with Power-law Activation}
\label{Sec: III(D)}
\begin{figure}                                              %fig6
\centering
\subfigure[WSN]{\label{fig: 6(a)}
\includegraphics[scale=0.21]{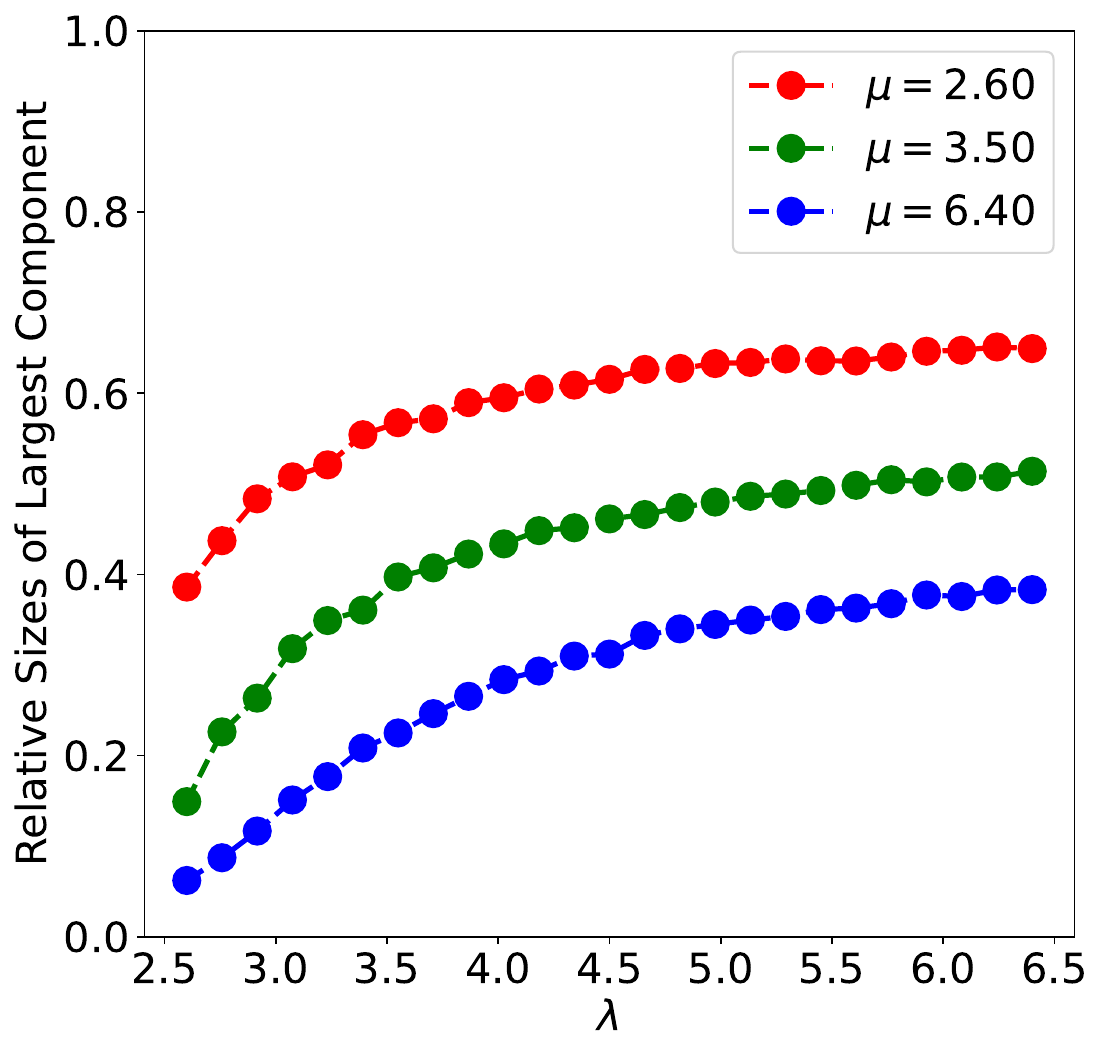}}
\subfigure[BAN]{\label{fig: 6(b)}
\includegraphics[scale=0.21]{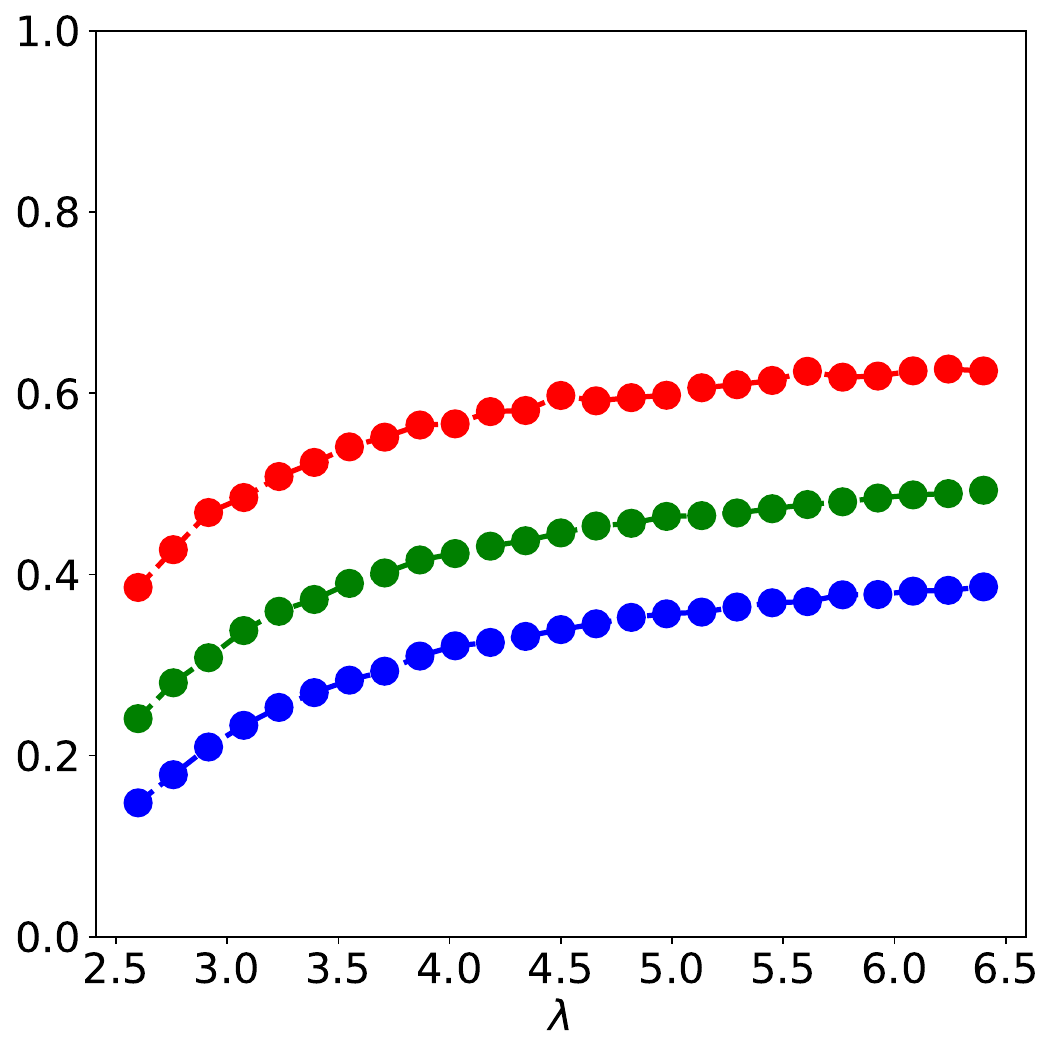}}
\vspace{-3.5mm}

\subfigure[Mouse Visual Cortex]{\label{fig: 6(c)}
\includegraphics[scale=0.21]{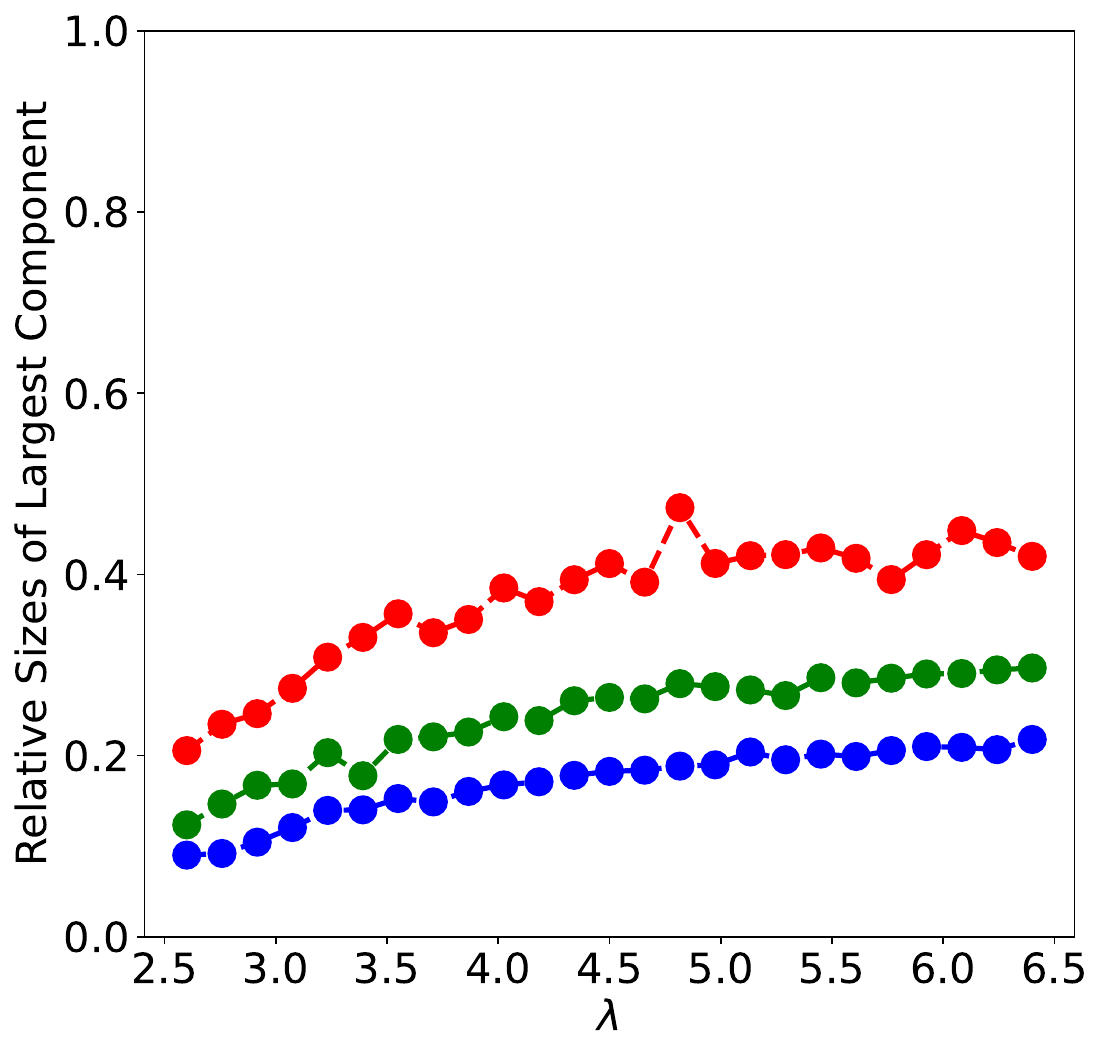}}
\subfigure[Infect Dublin]{\label{fig: 6(d)}
\includegraphics[scale=0.21]{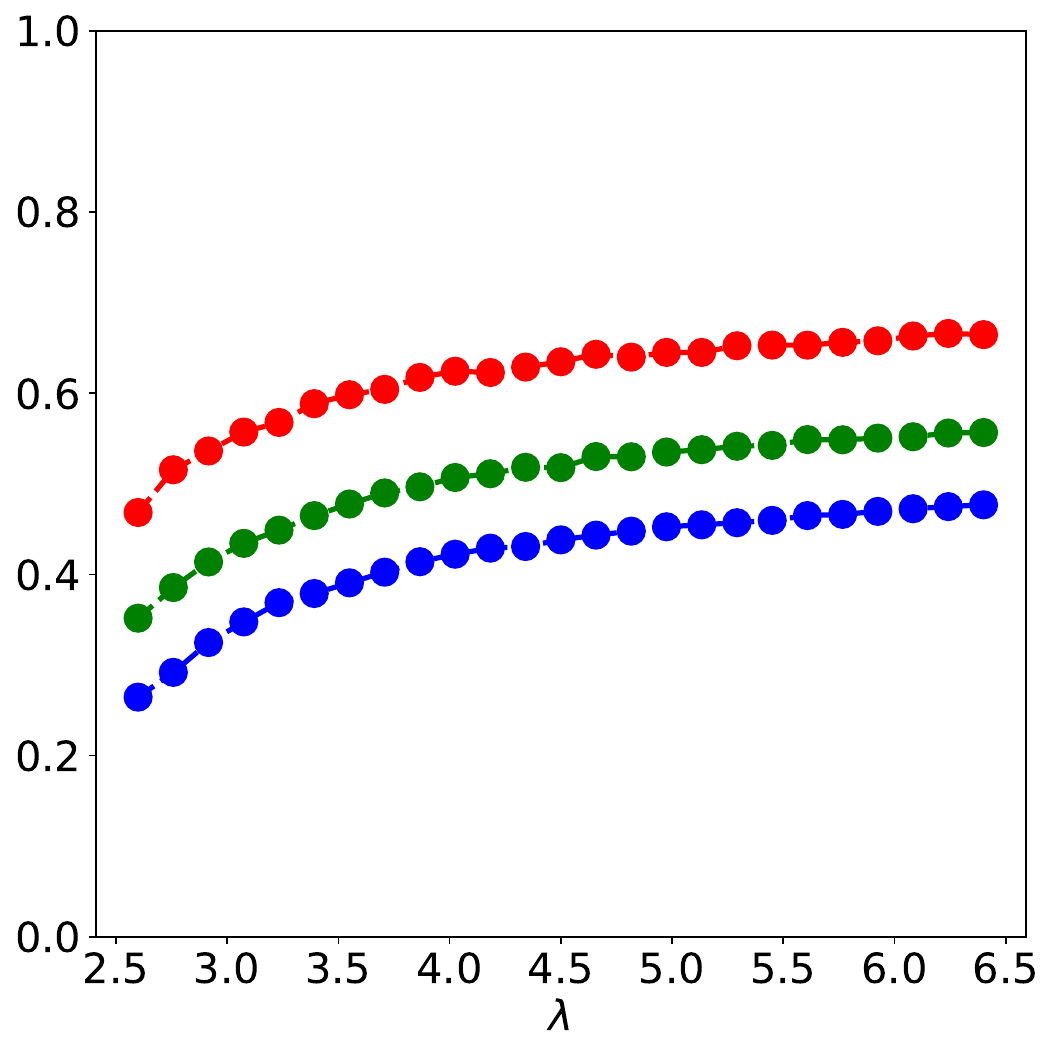}}
\vspace{-3.5mm}
\caption{\textbf{Relative sizes of the largest component.}  (a) WSN, $N=500$, $k=4$. (b) BAN, $N=500$, $k=4$. (c) Mouse visual cortex network, $N=193$. (d) Infect Dublin, infectious contact network, $N=410$. We set $\lambda\in[2.60,6.40]$ and $\mu\in\{2.60, 3.50, 6.40\}$, collect the relative sizes of the largest component in the time interval $t\in[50,150]$ and take the average for each data point. The relative size is the ratio of vertices in the largest component of the underlying network. (Color online)}\label{fig: 6}
\end{figure}
From Proposition \ref{proposition: 1}, an activated vertex is possible to be isolated. Additionally, the bridges in the network can also break, forming several components spontaneously due to the intermittent activation. Therefore, we can directly analyze the network resilience and robustness caused by the power-law activation patterns. In Fig. \ref{fig: 6}, we show the relative sizes of the largest component as functions of the activation rate $\lambda$. This quantity is calculated by the ratio of the vertices in the largest component of the activated subgraph. Due to the temporal links in the activated subgraph, we take the average of this quantity against time. Our results suggest that the relative size of the largest component grows as the increase of $\lambda$s, which has a positive correlation with the number of activated vertices suggested in Corollary \ref{corollary: 2}. 

From Figs. \ref{fig: 6(a)} and \ref{fig: 6(b)}, the relative sizes of WSNs are usually higher than ones in BANs. In the WSN, most vertices have degrees that are close to the expectation. Therefore, the probability for one single vertex to be isolated, as well as breaking the bridge in the network, is relatively lower. However, in the BAN, most vertices have minimal degrees, and hub vertices are connected with many neighbors with low clustering coefficients. Once the hub vertices are not activated, the center of the network breaks down and leaves several connected components. Accordingly, WSNs are likely to maintain more vertices in the largest components than BANs. The same is true in Figs. \ref{fig: 6(c)} and \ref{fig: 6(d)}. The mouse visual cortex network has more hub vertices and is more heterogeneous, leading to relatively weak robustness against the power-law activation patterns. 

\subsection{The Absorptivity of Evolutionary Dynamics}\label{Sec: III(E)}
\begin{figure}                                              %fig7
\centering
\subfigure[RRG, $C$ invade]{\label{fig: 7(a)}
\includegraphics[scale=0.21]{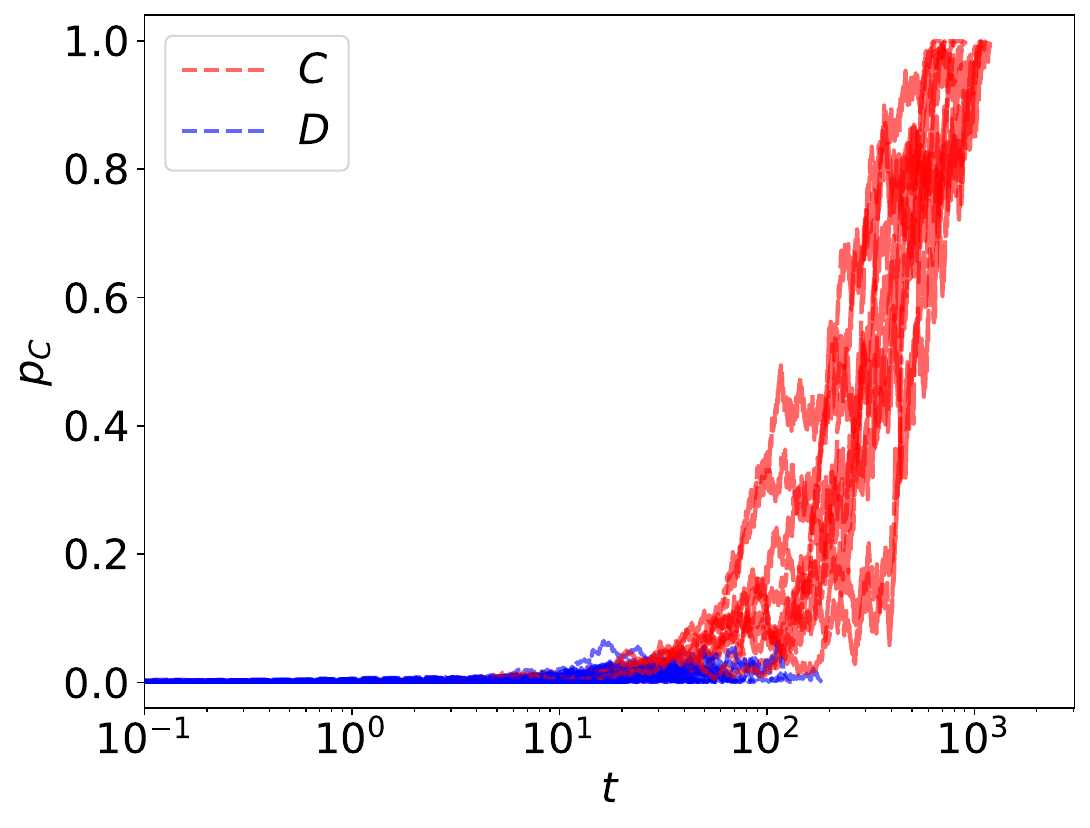}}
\subfigure[RRG, $D$ invade]{\label{fig: 7(b)}
\includegraphics[scale=0.21]{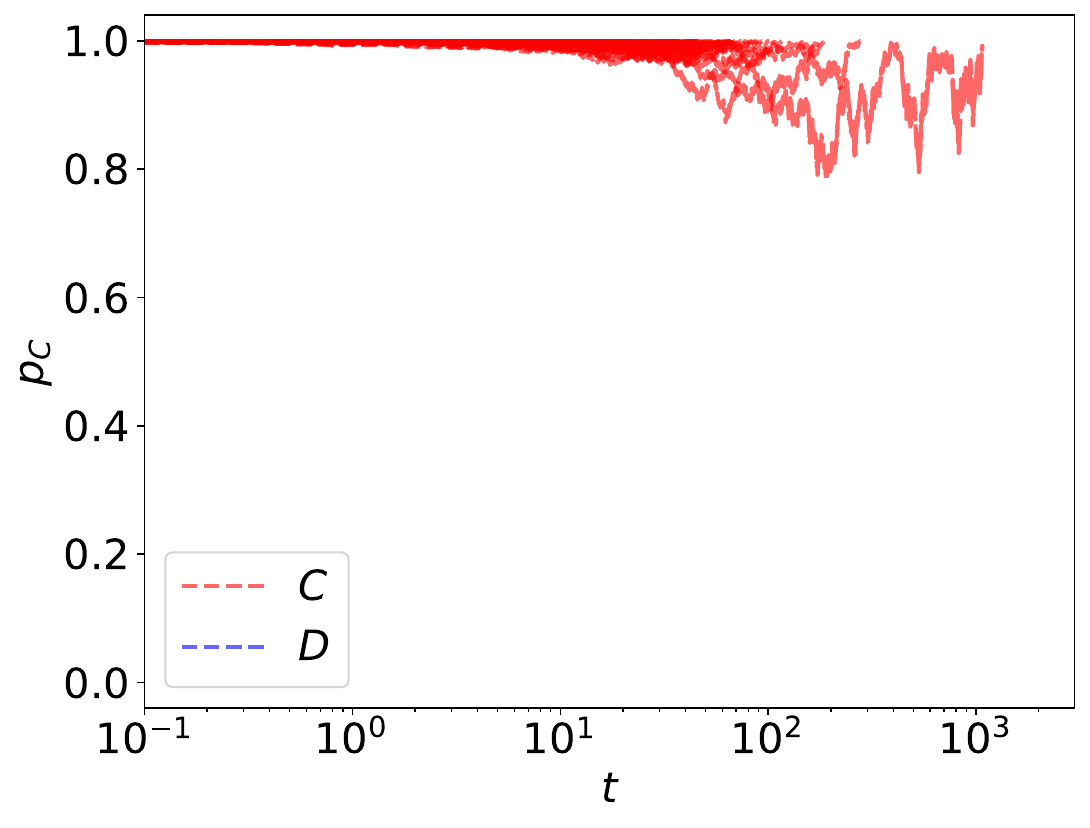}}
\vspace{-3.5mm}

\subfigure[WSN, $C$ invade]{\label{fig: 7(c)}
\includegraphics[scale=0.21]{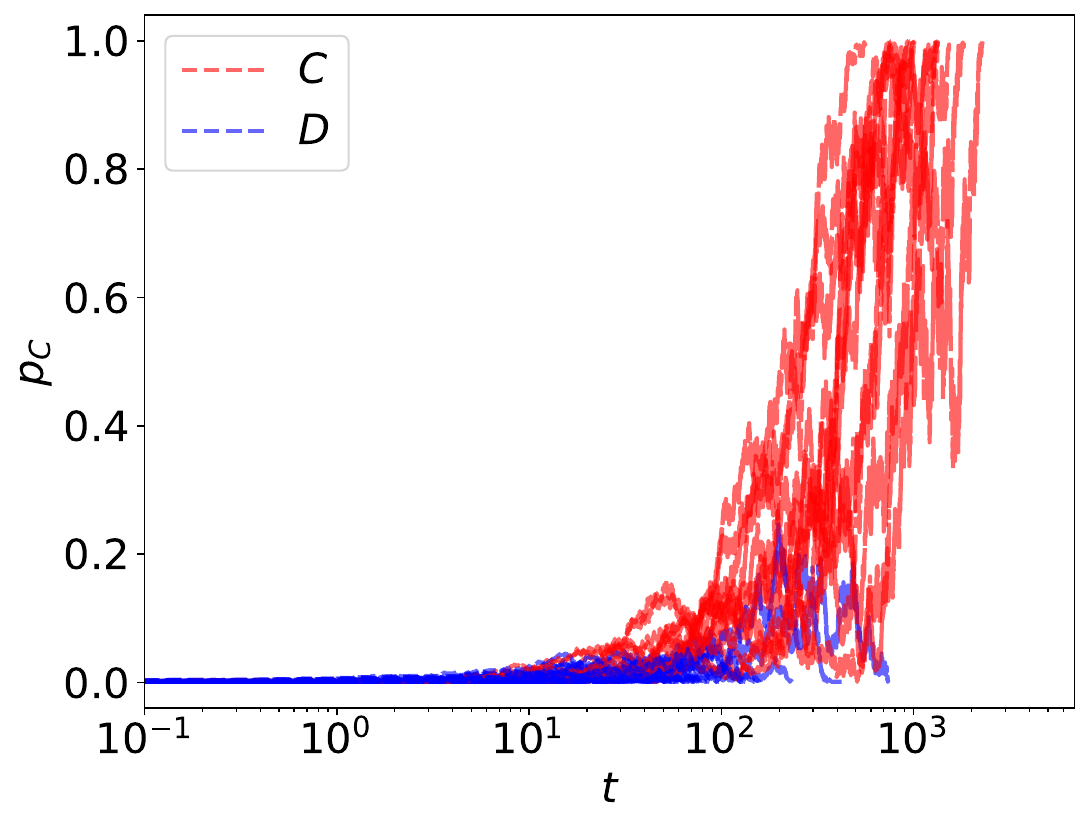}}
\subfigure[WSN, $D$ invade]{\label{fig: 7(d)}
\includegraphics[scale=0.21]{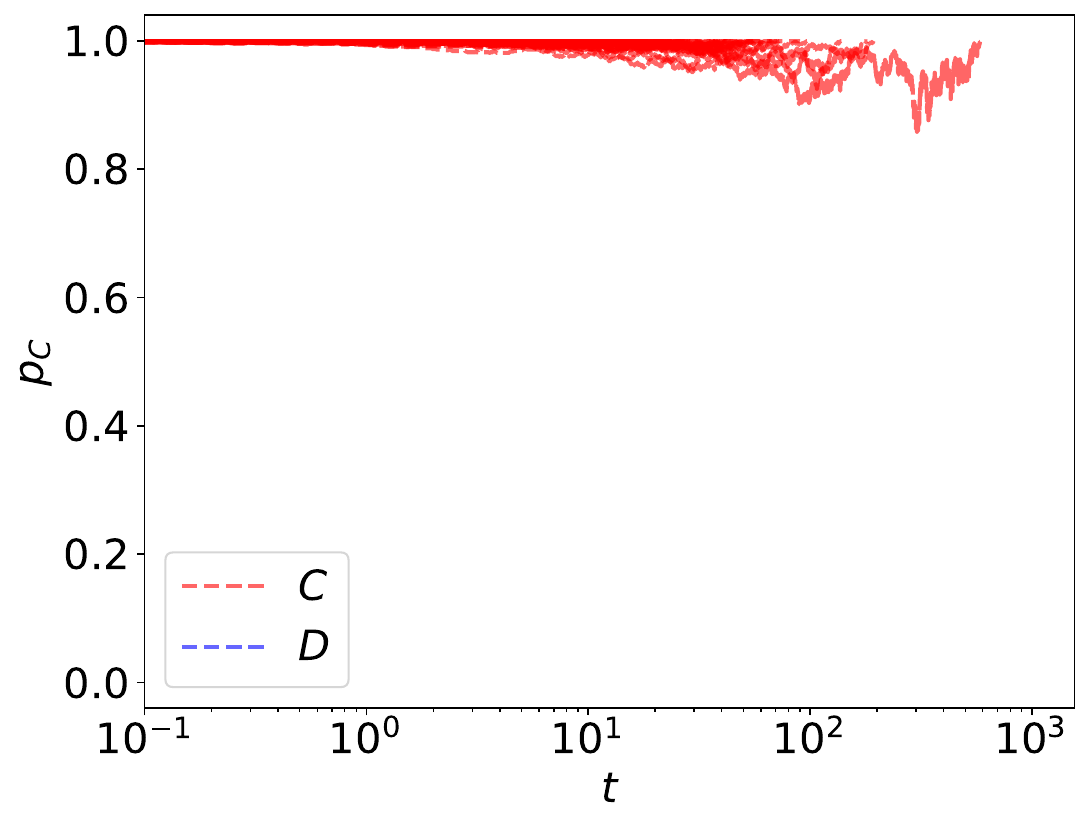}}
\vspace{-3.5mm}
\caption{\textbf{Cooperation density as the function of time. } The strategies of evolutionary dynamics evolve into pure absorb states without mutation. (a) RRG, cooperator invades in pure defection. (b) RRG, defector invades in pure cooperation. (c) WSN, cooperator invades in pure defection. (d) WSN, defector invades in pure cooperation. For parameters, we set $b=12$, $k=8$, $c=1$ for both RRG and WSN, and reconnection probability $0.4$ for WSN. The power-law activating patterns are with $\lambda=3.50$ and $\mu=3.70$. We plot the evolutionary trajectory of the cooperation frequency to demonstrate the absorptivity of the evolution process for $500$ times in each figure. We stop plotting until $p_C=1$ or $p_C=0$. Red and blue curves present the trajectory of the fixation of cooperation and defection respectively. (Color online)}\label{fig: 7}
\end{figure}
Before we further explore the fixation probability, we perform some repeated and independent experiments to show the absorptivity of the evolution processes, which guarantees that the strategy is absorbed instead of oscillating. We employ the RRG and WSN with $k=8$ and perform the simulation for cooperator invasion and defector invasion respectively. Results in Fig. \ref{fig: 7} show that the evolution processes converge to one of the pure cooperation and pure defection states. In this parameter setting, the cooperators are more likely to occupy the entire network than the defectors. Additionally, the absorption states are reached around $t=10^3$. 

\subsection{Evolution of Cooperation on Homogeneous Networks}\label{Sec: III(F)}
We employ the prisoner's dilemma game as stated in Sec. \ref{sec: II(B)} ($R=b-c$, $S=-c$, $T=b$, and $P=0$ with fixed $c=1$ and $b\in(1, 20]$) and consider network structures ($\mathcal{G}$) as RRG and WSN. For a better presentation, we denote present $\rho_C-\rho_D$ in the following experiments. Regarding other parameters, we fix $\lambda=3.50$ and $\mu=2.60$. Here we note that these sizes are sufficient for the study of fixation probability as \cite{ohtsuki2006simple}, \cite{su2019evolutionary}, and \cite{allen2017evolutionary} suggest. We set $k=[4, 8, 12, 16]$ and $N\in[100, 1000]$ for cross experiments. To approximate the fixation probability of cooperation, we randomly select one individual in the pure defective network as the invading cooperator for each group of parameters and run the evolutionary dynamics repeatedly and independently over $2\times10^3$ and $5\times10^4$ times for $N=1000$ and $100$ respectively. The frequency that cooperators take root in the population in the total experiment times is regarded as the fixation probability of the cooperation $\rho_C$. For the fixation probability of defection, we randomly select one individual as the invading defector in the pure cooperative population and perform similarly to approximate $\rho_D$. We have mentioned that $\rho_C>\rho_D$ is the condition for cooperation in a networked population. Therefore, we present $\rho_C-\rho_D$ as functions of $b$ to show the results in Fig. \ref{fig: 8}. 

As shown in Fig .\ref{fig: 8} and suggested in Proposition \ref{proposition: 5}, the power-law activating patterns can influence conditions for cooperation. This conclusion is accurate if the mean degree is small (or $N\gg k$) as shown in Fig. \ref{fig: 8} for each group of parameters. Concretely, the condition for cooperation is sound (sufficient and necessary) when $k=4$ and $8$ in the simulation. However, if we increase the degree to $k=12$ and $16$, this condition in Eq. \ref{eq: condition for cooperation} becomes sufficient but not necessary. The obtained data points and fitting lines are often already greater than zero if $(b/c)$ is greater than the critical value. Here, we note that errors exist caused by the small bias of simulations. With the increase of degree $k$, the power-law activating patterns increase the relaxation of critical cooperation conditions. Therefore, social systems with large average connections do not need the cost-to-benefit ratio equivalent to Eq. \ref{eq: condition for cooperation} to overcome social dilemmas with the power-law activating patterns.

In Fig. \ref{fig: 9}, we further consider mutation-selection dynamics in RRGs and WSNs. With mutation, the network does not reach the absorbing state, and cooperation is favored if the expected number of cooperators is higher than defectors, i.e., the frequency of cooperators is higher than $0.5$. As shown in Figs. \ref{fig: 9(a)} and \ref{fig: 9(b)}, we find that RRGs guarantee more cooperators in the population than WSNs, especially when the degree is relatively large. For example, if $k=4$, the conditions for cooperation in RRG and WSN are both around $b/c=5$. However, if $k=8$ and $12$, the critical conditions for cooperation in RRGs that we obtain by simulation are smaller than those in WSNs. Additionally, the exact solution of evolutionary dynamics with mutation is still an open issue. 
\begin{figure}                                              %fig8
\centering
\subfigure[RRG, $N=1000$]{\label{fig: 8(a)}
\includegraphics[scale=0.21]{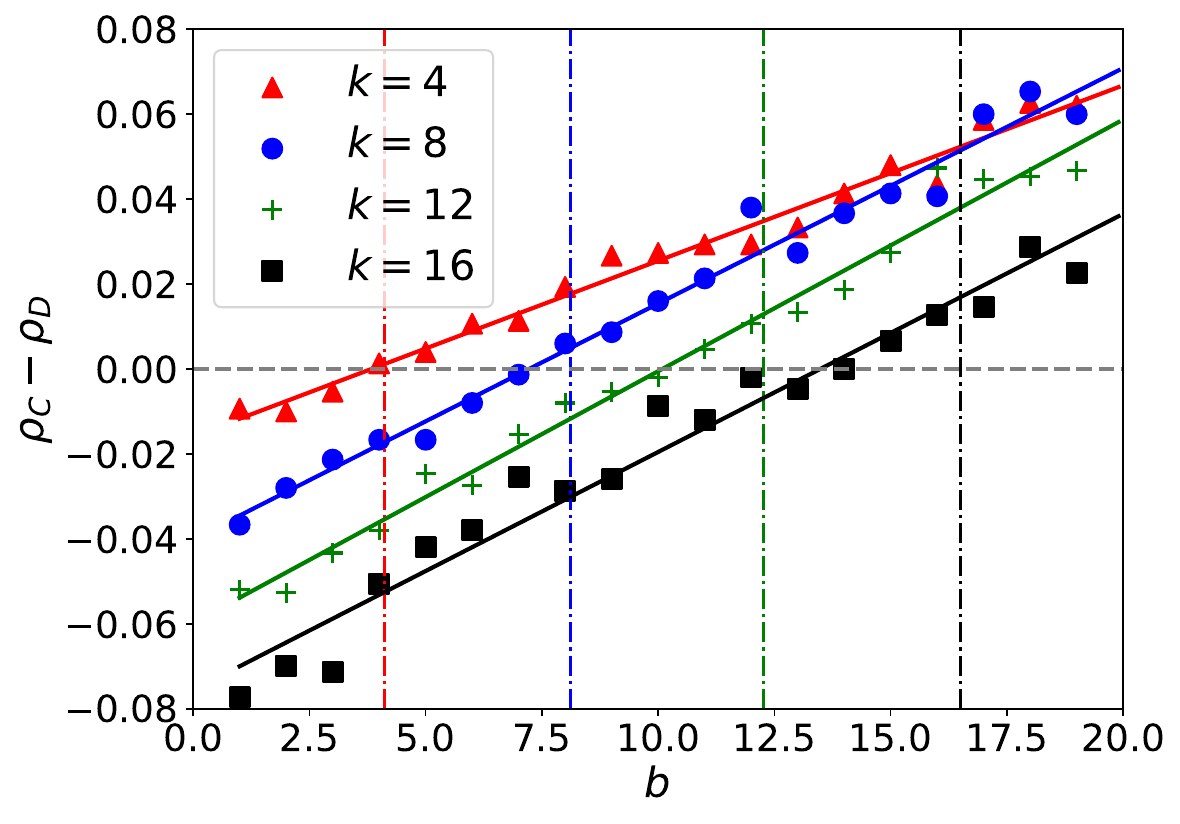}}
\subfigure[WSN, $N=1000$]{\label{fig: 8(b)}
\includegraphics[scale=0.21]{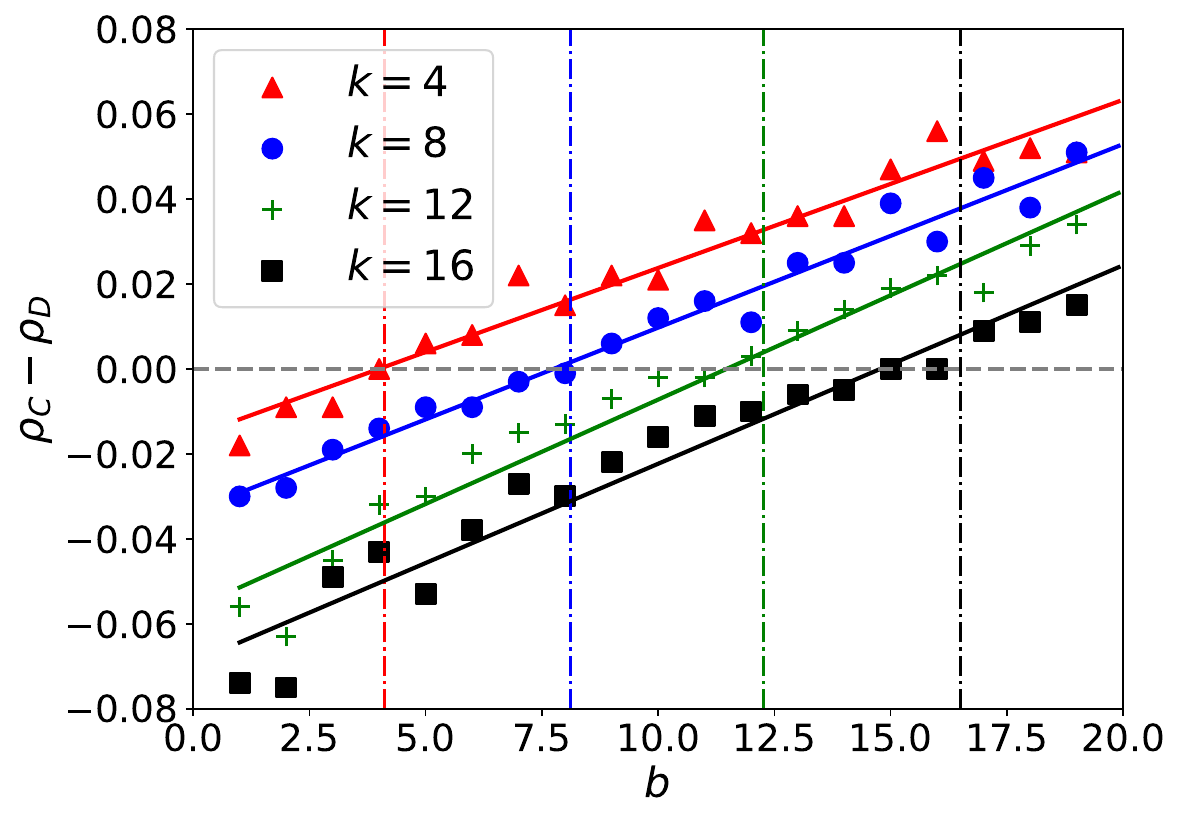}}
\vspace{-3.5mm}

\subfigure[RRG, $N=100$]{\label{fig: 8(c)}
\includegraphics[scale=0.21]{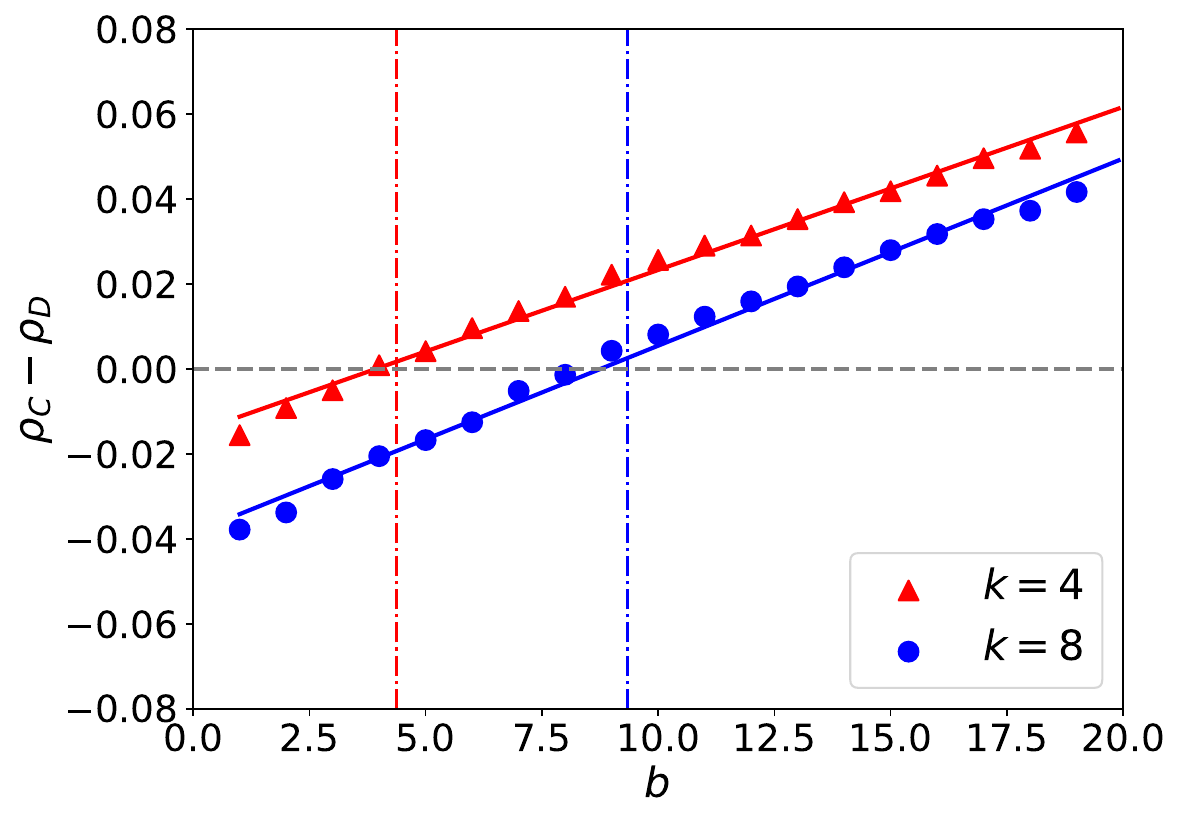}}
\subfigure[WSN, $N=100$]{\label{fig: 8(d)}
\includegraphics[scale=0.21]{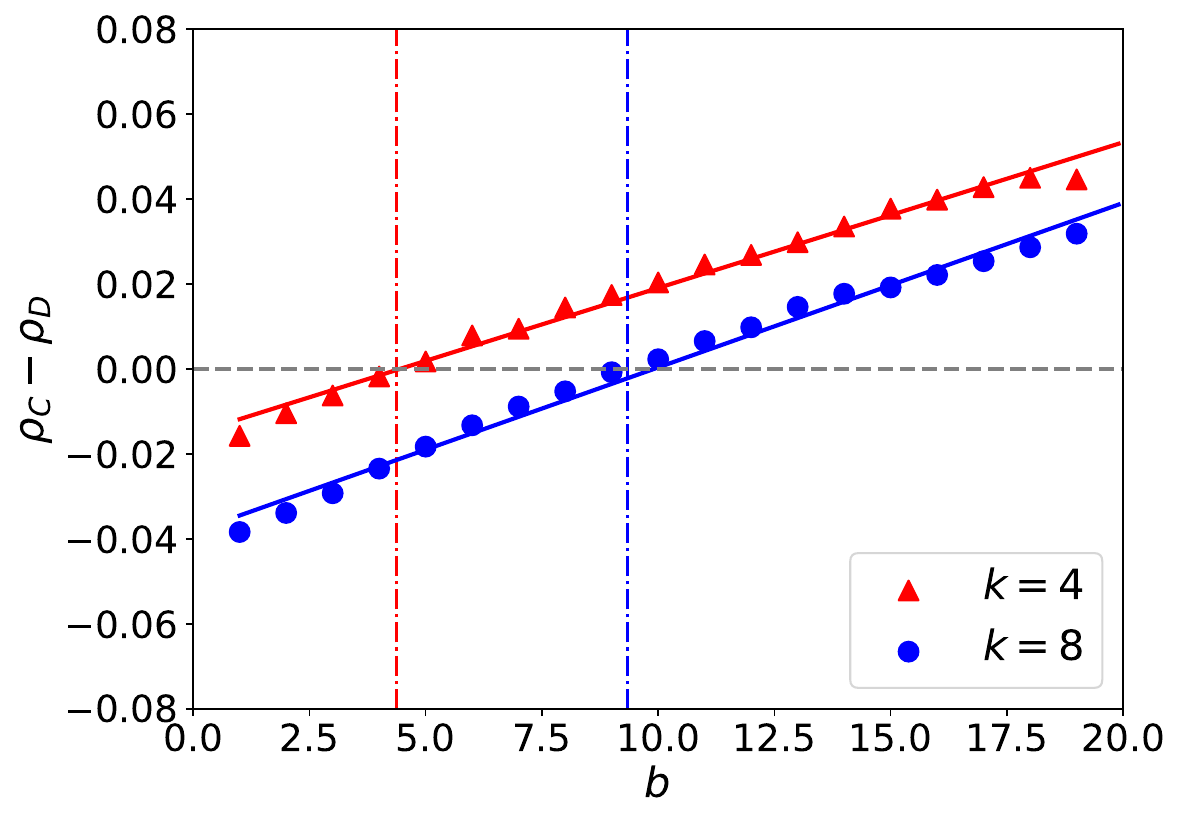}}
\vspace{-3.5mm}
\caption{\textbf{Fixation probabilities of the prisoner's dilemmas in homogeneous networks. } The theoretical condition for cooperation is sufficient but not necessary, especially for the large degree network. We fix $\lambda=3.50$, $\mu=2.60$, $c=1$. (a)-(b) show the results on RRGs and WSNs with $N=1000$ respectively, and (c)-(d) show the results for $N=100$. For variable parameters, we set $b\in(1, 20]$, and $k\in[4, 8, 12, 16]$. Red triangles, blue circles, green crosses, and black squares show the result for $k=4$, $8$, $12$, and $16$ respectively. Each data point is obtained by calculating the frequency of the final pure state in the evolution in the total $2\times 10^3$ and $5\times 10^4$ independent and repeated experiments for $N=1000$ and $N=100$ respectively. The solid lines in the same colors are the fitting lines to show the overall trends, which are obtained by linear regression. The dashed lines in the same colors present cooperation conditions for cooperation in Proposition \ref{proposition: 5}. The results show that our theorem holds for $N\gg k$ and becomes biased with the increase of degree. (Color online)}\label{fig: 8}
\end{figure}

\begin{figure}                                              %fig9
\centering
\subfigure[$RRG$]{\label{fig: 9(a)}
\includegraphics[scale=0.21]{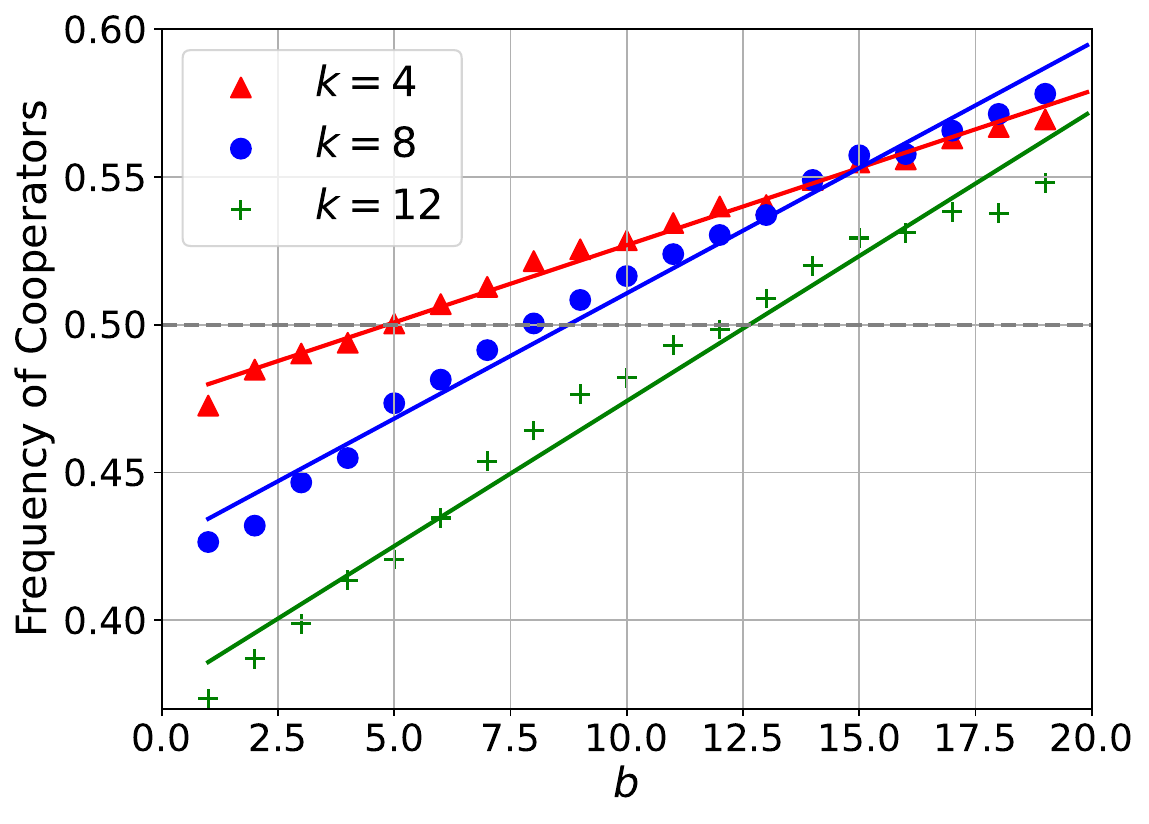}}
\subfigure[$WSN$]{\label{fig: 9(b)}
\includegraphics[scale=0.21]{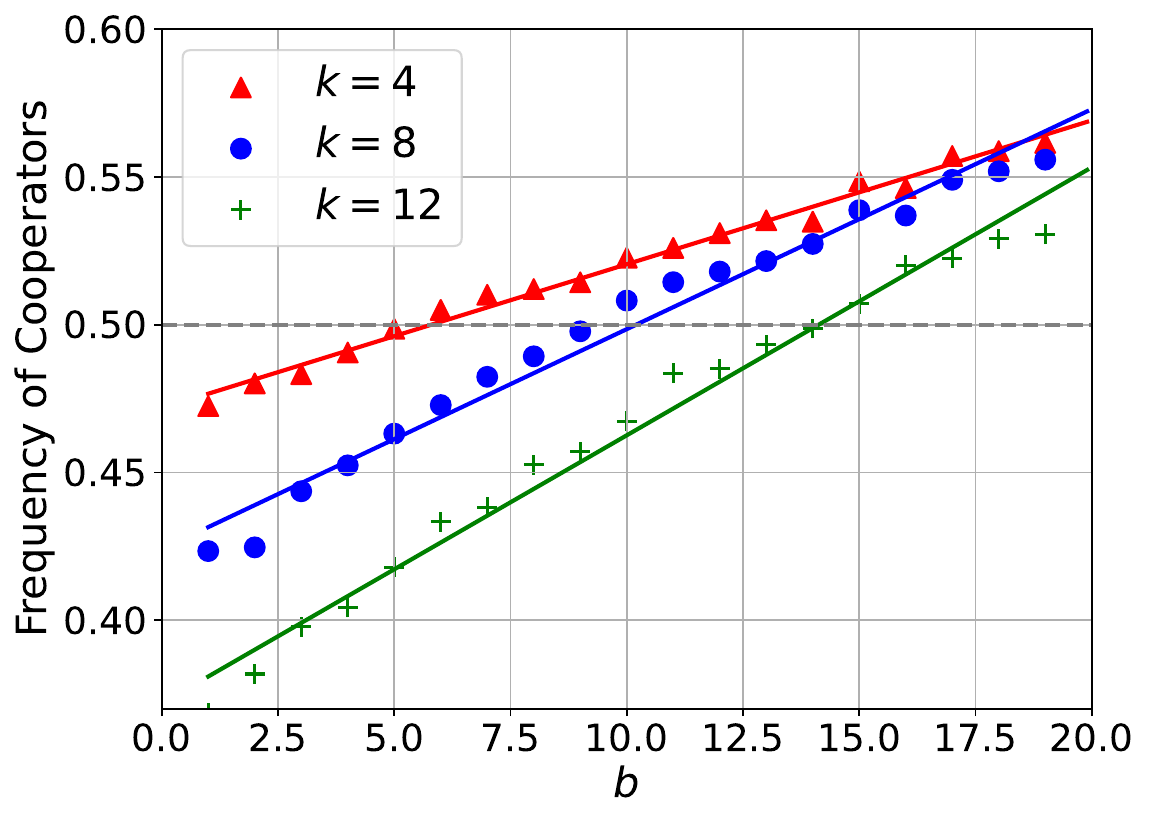}}
\vspace{-3.5mm}
\caption{\textbf{Stationary cooperation frequency of the prisoner's dilemmas in homogeneous networks with mutation. }RRGs promote cooperation with mutation greater than WSNs. In this figure, we set $N=1000$, $v=0.10$, $\lambda=3.50$ and $\mu=2.60$ in RRGs and WSNs. We fix $c=1$ for game parameters and observe the average cooperation frequency for $b\in(1, 20]$. Each data point is obtained by calculating the mean frequency in the total $10^4$ times after the stable state is reached. Red triangles, blue circles, and green crosses present the results for $k=4$, $8$, and $12$ respectively. The fitting lines are shown in the same colors, obtained by linear regression. (Color online)}\label{fig: 9}
\end{figure}

\subsection{Evolution of Cooperation on Real Networks}\label{Sec: III(G)}
\begin{figure}                                              %fig10
\centering
\subfigure[Dolphins]{\label{fig: 10(a)}
\includegraphics[scale=0.21]{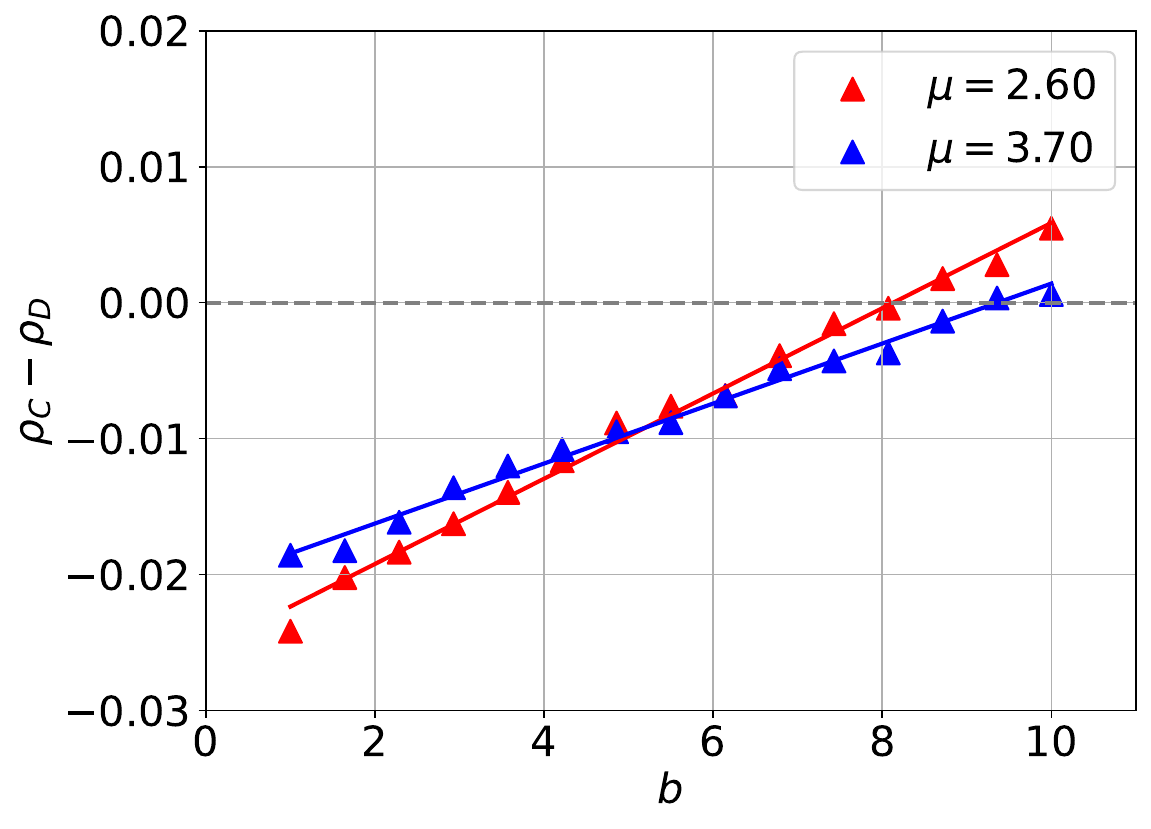}}
\subfigure[USA contiguous]{\label{fig: 10(b)}
\includegraphics[scale=0.21]{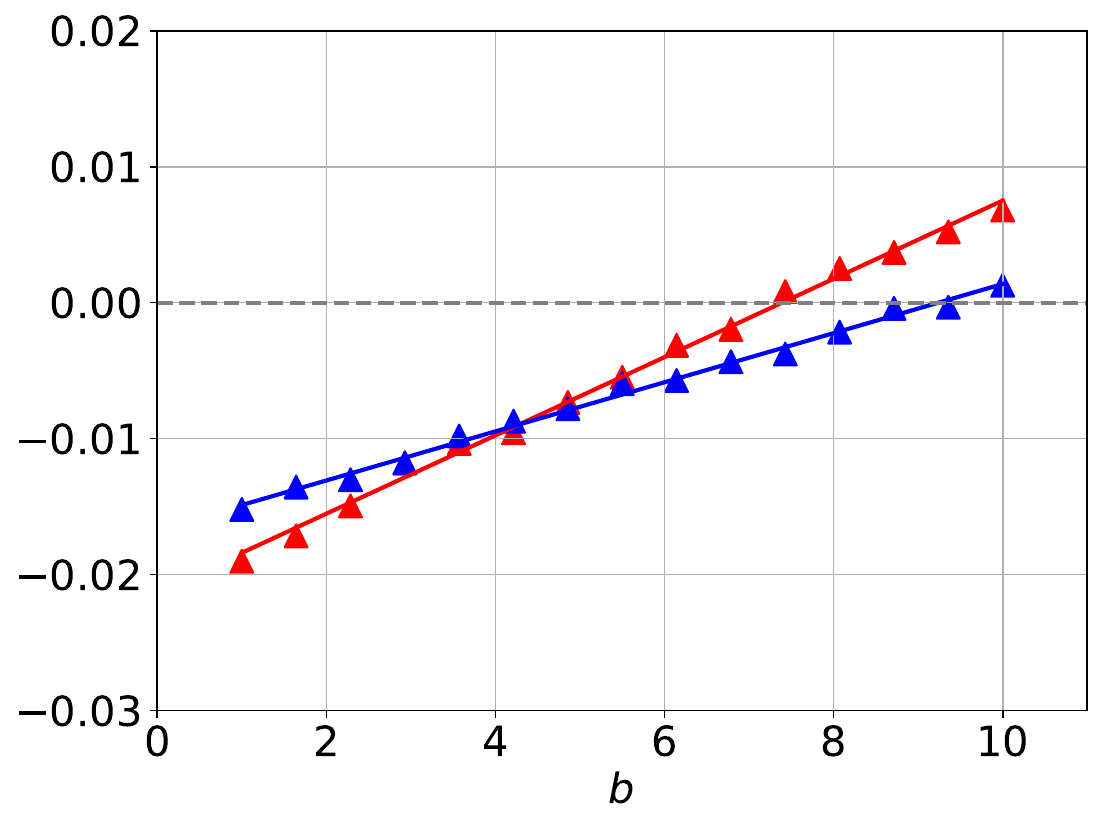}}
\vspace{-3.5mm}

\subfigure[Retweets]{\label{fig: 10(c)}
\includegraphics[scale=0.21]{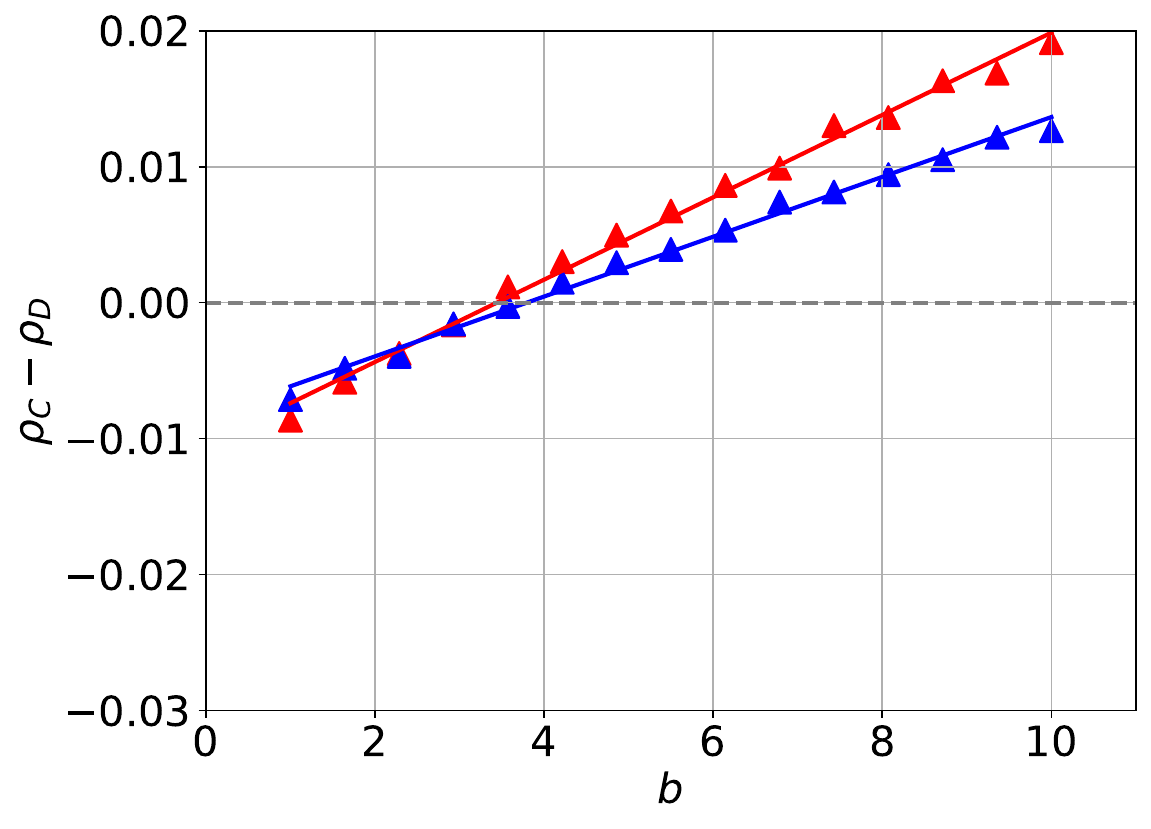}}
\subfigure[Sandia National Labs]{\label{fig: 10(d)}
\includegraphics[scale=0.21]{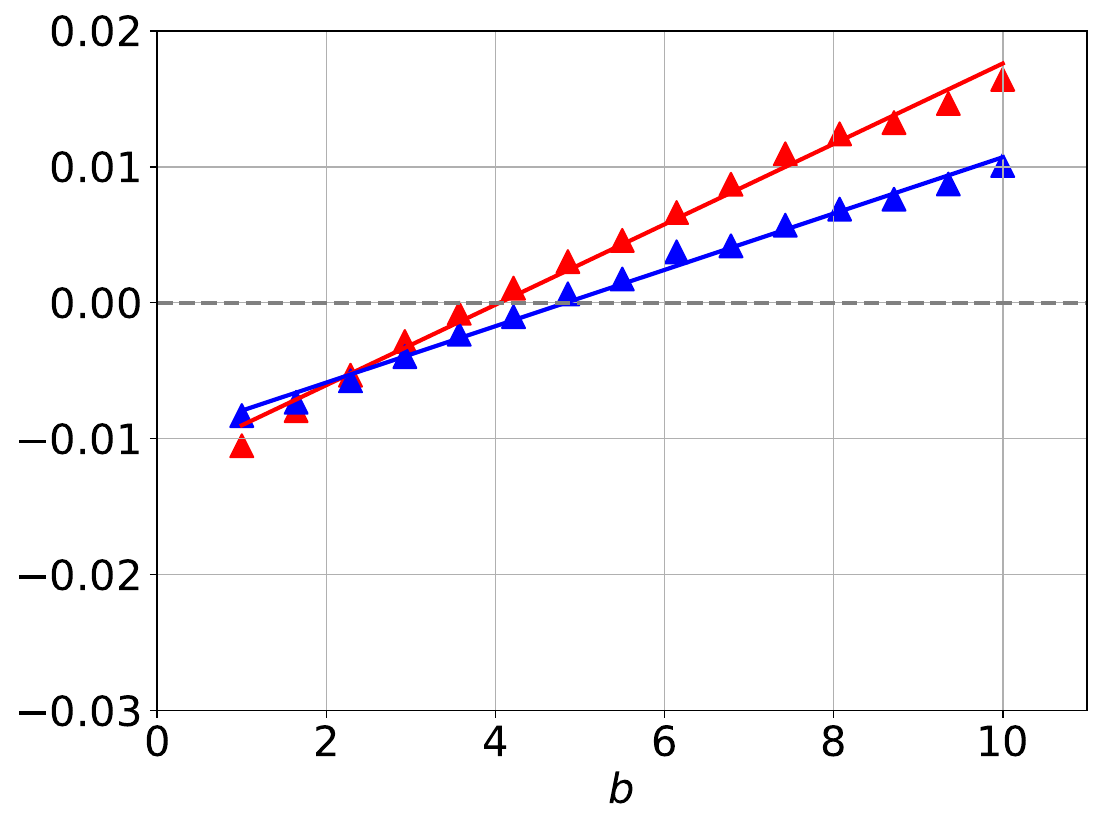}}
\vspace{-3.5mm}
\caption{\textbf{Fixation probabilities of the prisoner's dilemmas in real networks. } We fix $\lambda=3.50$, $\mu\in\{2.60,3.70\}$, and $c=1$. (a) The dolphin network, (b) the USA contiguous border network, (c) the Retweet network, and (d) the collaboration network of Sandia National Laboratories. For each curve, we show $\rho_C-\rho_D$ in $b\in(1,10]$ with $15$ data points. The results for $\mu=2.60$ and $\mu=3.70$ are presented in red and blue triangles respectively. The fixation probabilities for cooperation and defection are obtained by averaging $10^5$ independent experiments. The solid lines in the same colors are the fitting lines to show the overall trends, which are obtained by linear regression. (Color online)} \label{fig: 10}
\end{figure}

In addition to the generative networks, we use four real-world network data sets to further discuss the effect of power-law activating patterns on the fixation of cooperation. Our employed real networks are \cite{nr-aaai15}: i) the dolphin network, the interaction relationship of 62 dolphins, $\vert\mathcal{E}\vert=159$, ii) the USA contiguous network, the border of 49 states in America, $\vert\mathcal{E}\vert=107$, iii) the retweet network, retweets among 96 Twitter users, $\vert\mathcal{E}\vert=117$, and iv) the collaboration network of 86 scientists at Sandia National Laboratories $\vert\mathcal{E}\vert=124$. Using these four networks as the underlying network $\mathcal{G}$, we can study the effect of power-law intermittent interaction among individuals on the evolution dynamics of cooperation. 

From the results in Fig. \ref{fig: 10}, we find that $\mu=2.60$ can yield a less strict condition for the fixation of cooperation than $\mu=3.70$. The cost that each individual should pay to maintain the cooperation is lower for $\mu=2.60$. Additionally, our focal quantity $\rho_C-\rho_D$ increases faster with the growth of $b$ if $\mu=2.60$. According to Proposition \ref{proposition: 1}, the increase of $\mu$ leads to a greater probability of finding a vertex in the quiescent state, which causes fewer neighbors to interact with on average. Namely, a large mean degree of the activated subgraph is beneficial to the fixation of cooperation, which is against previous findings on the condition for cooperation based on network degrees \cite{ohtsuki2006simple,konno2011condition,su2019evolutionary}. A similar trend is also shown in Proposition \ref{proposition: 5}, but the exact solution for the success of cooperation on any network structure is still an open issue in this paper. The intermittent interaction of individuals may be the reason for real-life populations to overcome social dilemmas and achieve cooperation. 
\section{Conclusion}\label{Sec: IV}
\small
In this article, we emphasize the power-law inter-event time of vertices in complex networks and propose power-law activating patterns based on continuous Markov chain theory. In the proposed model, we assume that each vertex in the network switches between two states (activated and quiescent) with power-law rates, i.e., each vertex stays in each state for a random time that follows two independent power-law distributions. We find a homogeneous stationary distribution of activated vertex numbers and give some statistical properties. To examine the dynamic process in the proposed network model, we employ the two-person-two-strategy evolutionary game theory and study the fixation probability and critical cooperation condition. We show that the strategic combination of the population evolves to one of the pure cooperative or defective states, and the power-law activating patterns can influence both the fixation probability and the cooperation condition.

Our most significant findings in this paper are the homogeneous stability of network sizes with heterogeneous activity patterns and its positive effect on cooperation in evolutionary dynamics. We note that this topic can be further studied by importing more general assumptions to fit the real world. For example, the power-law distribution of the inter-event time is the most common one of the many situations. Exponential distributions, power-law distributions with exponential cut-off, and log-normal distributions are also found in real-world data sets. Our study on power-law activating patterns can be extended to any probability distribution. Additionally, the correlation of activation among vertices can lead to different results on both network topology and evolutionary dynamics. As we have discussed, the cooperation condition that we obtained theoretically is sufficient but not necessary. The precise (sufficient and necessary conditions) critical condition for the cooperation of the evolutionary dynamics is still not clear in this paper. In addition to the evolutionary dynamics, other dynamic processes (e.g., epidemic propagation, synchronization, and percolation) are vital in understanding the effects of network structures. Therefore, studying vertex phase switching on these dynamic processes is essential to further capture the effect of inter-event interaction patterns. 
\small
\bibliographystyle{ieeetr}
\small
%\bibliography{refs}

\ifCLASSOPTIONcaptionsoff
  \newpage
\fi
\begin{IEEEbiography}[{\includegraphics[width=1in,height=1.25in,clip,keepaspectratio]{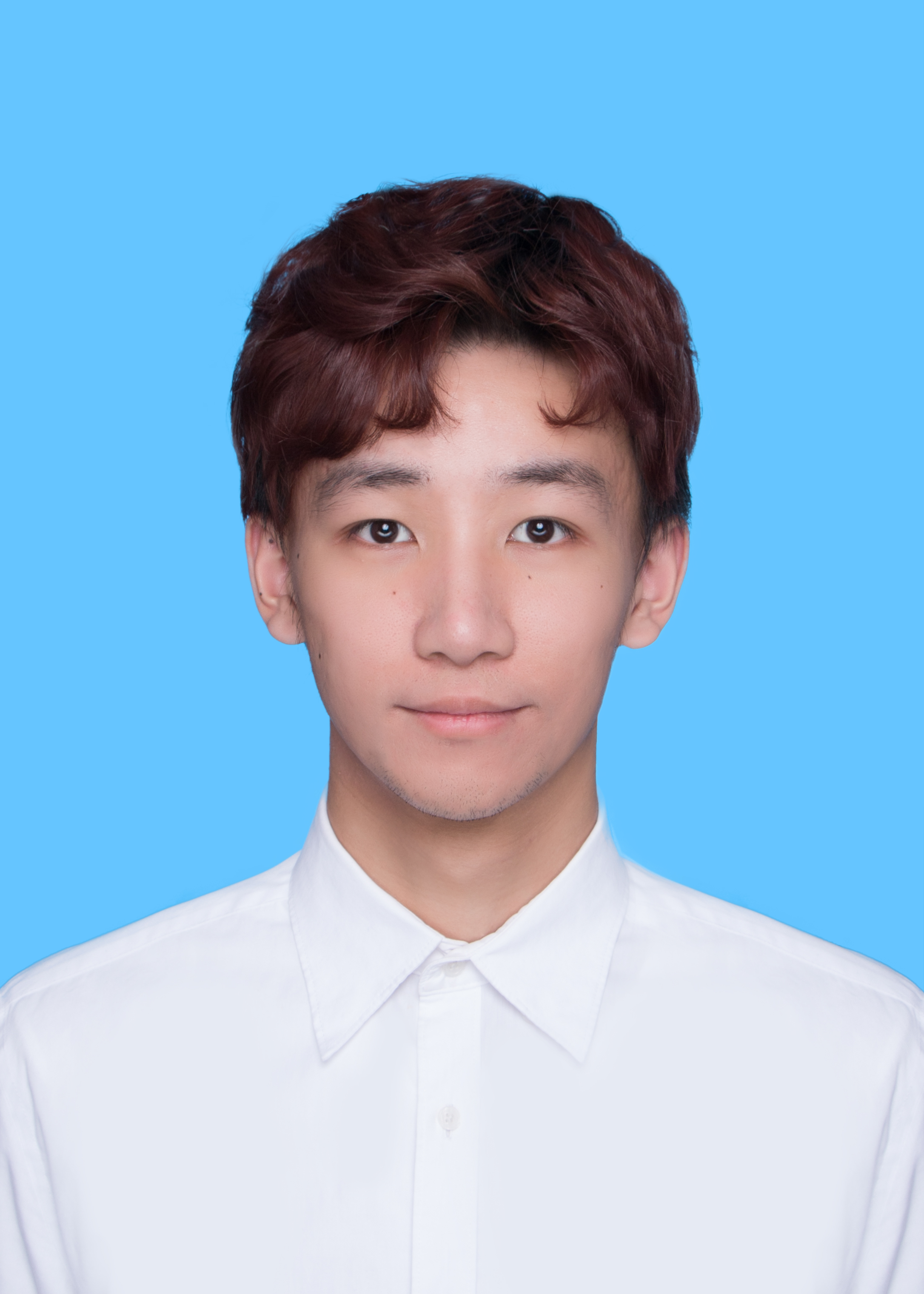}}]{Ziyan Zeng}
received the bachelor’s degree from the College of Artificial Intelligence, Southwest University, Chongqing, China. He is currently pursuing the master’s degree in computer science. His research interests include complex networks, evolutionary games, stochastic processes, and nonlinear science. 
\end{IEEEbiography}
\begin{IEEEbiography}
[{\includegraphics[width=1in,height=1.25in,clip,keepaspectratio]{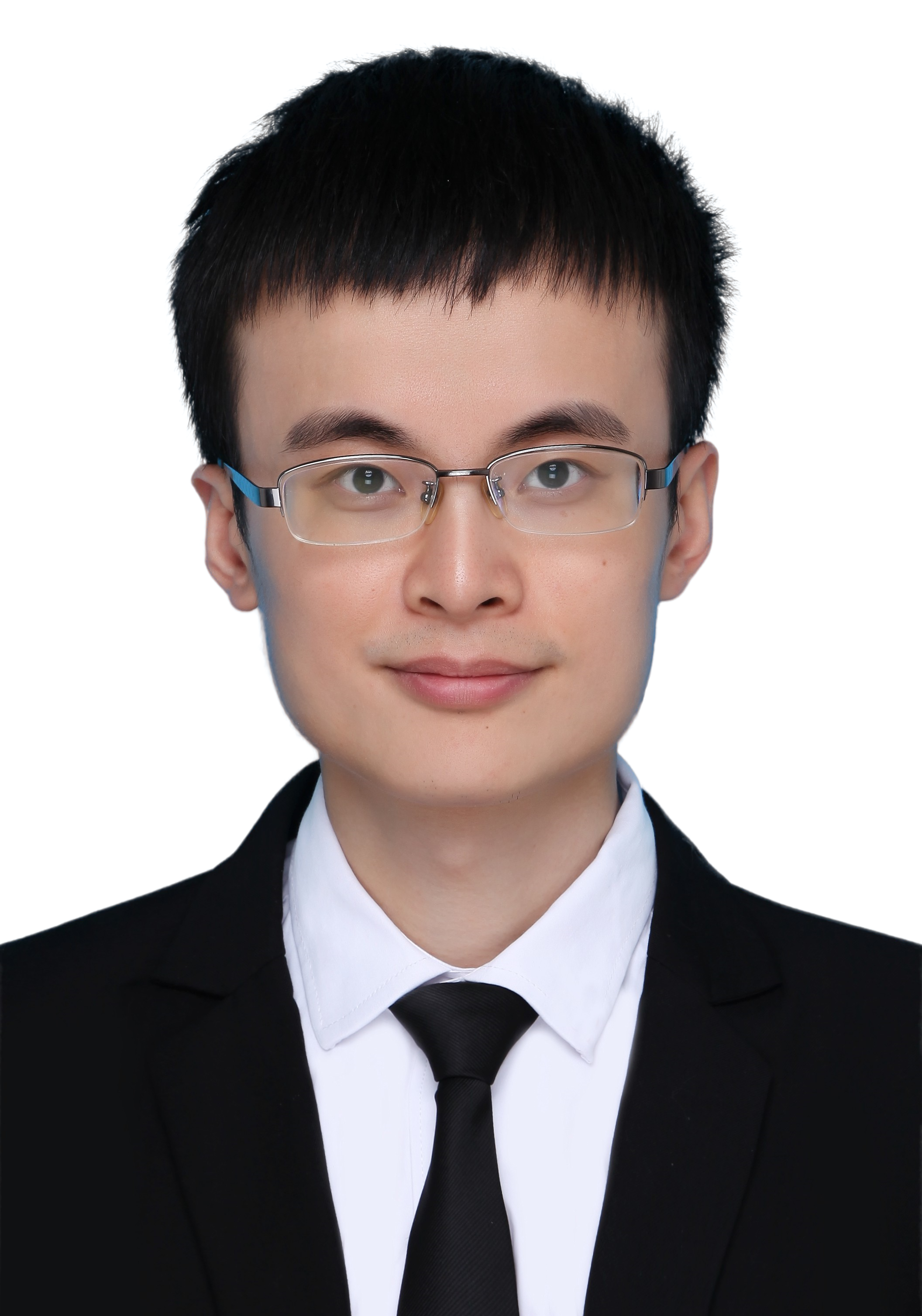}}]{Minyu Feng}
received his Ph.D. degree in Computer Science from a joint program between University of Electronic Science and Technology of China, Chengdu, China, and Humboldt University of Berlin, Berlin, Germany in 2018. Since 2019, he has been an associate professor at the College of Artificial Intelligence, Southwest University, Chongqing, China. 

Dr. Feng is a Senior Member of IEEE, China Computer Federation (CCF), and Chinese Association of Automation (CAA). Currently, he serves as a Subject Editor of Applied Mathematical Modelling and an Editor of International Journal of Mathematics for Industry. He is also a Reviewer for Mathematical Reviews of the American Mathematical Society.

Dr. Feng's research interests include Complex Systems, Evolutionary Game Theory, Computational Social Science, and Mathematical Epidemiology.
\end{IEEEbiography}
\begin{IEEEbiography}
[{\includegraphics[width=1in,height=1.25in,clip,keepaspectratio]{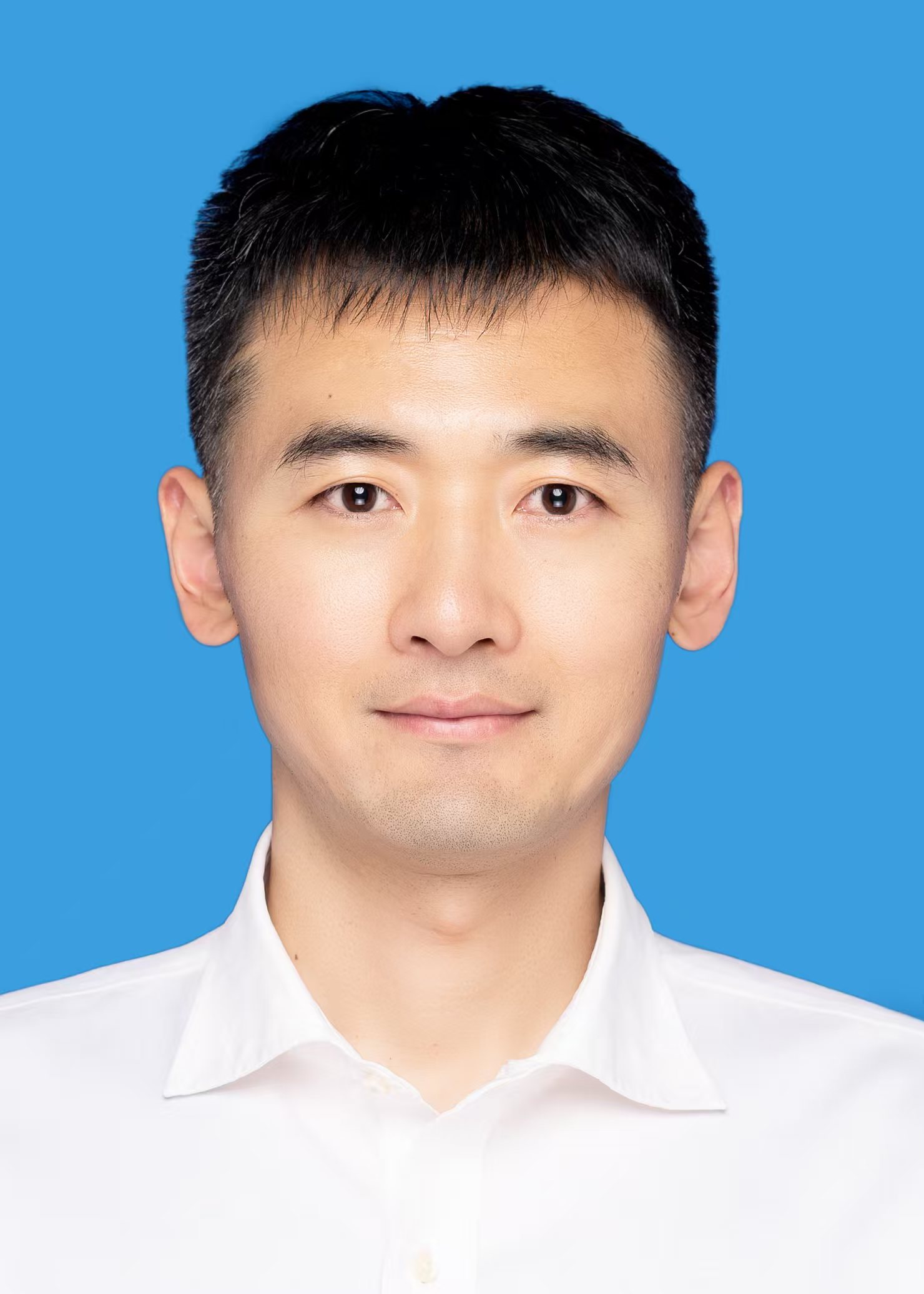}}]{Pengfei Liu}
received the B.Sc., M.Sc., and Ph.D. degrees from University of Electronic Science and Technology of China (UESTC) in 2009, 2012, and 2018, respectively, all in computer science and technology. He is an assistant researcher with UESTC and a Post-Doctoral Researcher with  Yangtze Delta Region Institute (Quzhou), UESTC. His current research interests include complex systems, multi-agent systems, swarm intelligence and edge computing.
\end{IEEEbiography}
\begin{IEEEbiography}
[{\includegraphics[width=1in,height=1.25in,clip,keepaspectratio]{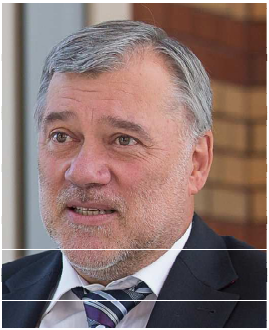}}]{J\"{u}rgen Kurths}
received the B.S. degree in mathematics from the University of Rostock, Rostock, Germany, the Ph.D. degree from the Academy of Sciences, German Democratic Republic, Berlin, Germany, in 1983, the Honorary degree from N.I.Lobachevsky State University, Nizhny Novgorod, Russia, in 2008, and the Honorary degree from Saratow State University, Saratov, Russia, in 2012.

From 1994 to 2008, he was a Full Professor with the University of Potsdam, Potsdam, Germany. Since 2008, he has been a Professor of Nonlinear Dynamics with the Humboldt University of Berlin, Berlin, Germany, and the Chair of the Research Domain
 Complexity Science with the Potsdam Institute for Climate Impact Research,
 Potsdam. He has authored more than 700 papers, which are cited more than
 60000 times (H-index: 111). His main research interests include synchronization, complex networks, time series analysis, and their applications.
 
 Dr. Kurths was the recipient of the Alexander von Humboldt Research
 Award from India, in 2005, and from Poland in 2021, the Richardson Medal of
 the European Geophysical Union in 2013, and the Eight Honorary Doctorates.
 He is a Highly Cited Researcher in Engineering. He is a member of the
 Academia 1024 Europaea. He is an Editor-in-Chief of Chaos and on the
 editorial boards of more than ten journals. He is a Fellow of the American
 Physical Society, the Royal Society of Edinburgh, and the Network
 Science Society. 
\end{IEEEbiography}
\end{document}